\documentclass[11pt]{article}
\usepackage[paper=a4paper,
left=2.5cm, right=2.5cm,
top=2cm, bottom=2.5cm,
footskip=0.75cm
]{geometry}
\usepackage[utf8]{inputenc}

\usepackage[usenames,dvipsnames,svgnames,table]{xcolor}

\usepackage{amsmath}
\usepackage{amsthm} 
\usepackage{amssymb}
\usepackage{amsfonts}
\usepackage{bbm} 
\usepackage[linktocpage=true,ocgcolorlinks,bookmarks,bookmarksopen,bookmarksnumbered]{hyperref}
\hypersetup{colorlinks,linkcolor=cyan!60!black,citecolor=ForestGreen,urlcolor=magenta!50!black}
\usepackage{subcaption, graphicx, multirow, tabularx}
\usepackage{mathtools}
\usepackage{mathrsfs} 
\usepackage[all]{xy}
\usepackage{bm}

\usepackage{stmaryrd} 

\usepackage{float} 

\graphicspath{{./}}

\usepackage{thmtools}

\newcommand{\idea}{\noindent{\itshape Proof idea. }}

\newcommand{\formatproblem}[1]{\mathbf{#1}}
\newcommand{\EQ} {\formatproblem{EQ}} 
\newcommand{\EQout} {\formatproblem{EQ}^{\out}} 
\newcommand{\GHD} {\formatproblem{GHD}} 
\newcommand{\GT} {\formatproblem{GT}} 
\newcommand{\MAX} {\formatproblem{MAX}} 

\newcommand{\id} {\formatproblem{id}} 
\newcommand{\CondId} {\formatproblem{CondId}} 
\newcommand{\SplitId} {\formatproblem{SplitId}}   
\newcommand{\DISJ} {\formatproblem{DISJ}} 
\newcommand{\INT} {\formatproblem{INT}} 
\newcommand{\FtFD} {\formatproblem{FtFD}} 
\newcommand{\XOR} {\formatproblem{XOR}}

\newcommand{\GapMAJ}[1][]{\formatproblem{GapMAJ}_{#1}}
\newcommand{\GapMajX}[1][]{\GapMAJ[#1]{\circ}\XOR} 

\newcommand{\TDE} {\formatproblem{TDE}} 
\newcommand{\weavesymbol}{{\mathsf{x\!x}}}
\newcommand{\weave}{\mathbin{\weavesymbol}}

\newcommand{\indic}{\mathbf{1}}

\DeclareMathOperator{\Exp}{\ensuremath{\mathbb{E}}}

\newcommand{\transpose}[1]{\prescript{t}{}{#1}}

\newcommand{\zo}{\set{0,1}}

\DeclareMathOperator{\prt}{prt} 
\DeclareMathOperator{\wprt}{wprt} 
\DeclareMathOperator{\rank}{rank} 

\DeclarePairedDelimiter{\abs}{\lvert}{\rvert}
\DeclarePairedDelimiter{\card}{\lvert}{\rvert}
\DeclarePairedDelimiter{\set}{\lbrace}{\rbrace}
\DeclarePairedDelimiter{\event}{\lbrack}{\rbrack}
\DeclarePairedDelimiter{\range}{\lbrack}{\rbrack}
\DeclarePairedDelimiter{\parens}{\lparen}{\rparen}
\DeclarePairedDelimiter{\floor}{\lfloor}{\rfloor}
\DeclarePairedDelimiter{\ceil}{\lceil}{\rceil}

\newcommand{\CC}{{\operatorname{CC}}}

\newcommand{\formatscript}[1]{\mathsf{#1}}
\newcommand{\pub}{\formatscript{pub}}
\newcommand{\priv}{\formatscript{priv}}

\newcommand{\diff}{\formatscript{diff}}
\newcommand{\out}{\formatscript{out}} 
\newcommand{\lf}{\formatscript{lf}} 

\newcommand{\op}{\formatscript{open}} 
\newcommand{\loc}{\formatscript{loc}} 
\newcommand{\ali}{\formatscript{A}} 
\newcommand{\bob}{\formatscript{B}} 
\newcommand{\uni}{\formatscript{uni}} 
\newcommand{\oot}{\formatscript{1of2}} 
\newcommand{\spt}{\formatscript{spl}} 
\newcommand{\xor}{\formatscript{xor}} 
\newcommand{\mdl}{\mathcal{M}} 

\newcommand{\cM} {\mathcal M} 
\newcommand{\cO} {\mathcal O} 
\newcommand{\cN} {\mathcal N} 
\newcommand{\cP} {\mathcal P}
\newcommand{\cR} {\mathcal R} 

\newcommand{\cT} {\mathcal T} 
\newcommand{\cU} {\mathcal U}
\newcommand{\cV} {\mathcal V}
\newcommand{\cX} {\mathcal X} 
\newcommand{\cY} {\mathcal Y} 
\newcommand{\cZ} {\mathcal Z} 

\newcommand{\bbR}{{\mathbb R}}
\newcommand{\bbN}{{\mathbb N}}
\newcommand{\bbF}{{\mathbb F}}

\newtheorem{theorem}{Theorem}[section]  
\newtheorem{definition}[theorem]{Definition}
\newtheorem{lemma}[theorem]{Lemma}
\newtheorem{proposition}[theorem]{Proposition}
\newtheorem{corollary}[theorem]{Corollary}
\newtheorem{example}[theorem]{Example}
\usepackage[capitalise,nameinlink]{cleveref}
\crefname{point}{Point}{Points}
\crefname{step}{Step}{Steps}
\crefname{case}{case}{cases}
\crefname{subproposition}{subproposition}{subpropositions}

\newcolumntype{Y}{>{\centering\arraybackslash}X}
\newcolumntype{Z}[1]{>{\setlength\hsize{#1\hsize}\centering\arraybackslash}X}
\newcolumntype{V}[1]{>{\setlength\hsize{#1\hsize}\arraybackslash}X}

\begin{document}

\title{The  communication complexity of functions with large outputs}

\newcommand{\email}[1]{\href{mailto:#1}{\texttt{#1}}}

\author{Lila Fontes%
\thanks{Swarthmore College -- \email{fontes@cs.swarthmore.edu}
}
\and Sophie Laplante%
\thanks{IRIF, Universit\'e Paris Cité -- \email{laplante@irif.fr}
}
\and Mathieu Lauri{\`e}re%
\thanks{NYU-ECNU Institute of Mathematical Sciences, NYU Shanghai -- \email{ml5197@nyu.edu}
}
\and Alexandre Nolin%
\thanks{CISPA Helmholtz Center for Information Security -- \email{alexandre.nolin@cispa.de}
}}
\date{}
\maketitle

\begin{abstract}
We study the two-party communication complexity of functions with large outputs,
and show that the communication complexity can greatly vary depending on
what output model is considered.
We study a variety of output models, ranging from the \emph{open model}, in which an external observer can compute the outcome, to the \emph{XOR model}, in which the outcome of the protocol should be the bitwise XOR of the players' local outputs. This model is inspired by XOR games, which are widely studied two-player quantum games.

We focus on the question of error-reduction in these new output models.
For functions of output size~$k$, applying standard error
reduction techniques in the XOR model would introduce an additional cost linear in $k$. We
show that no dependency on $k$ is necessary. 
Similarly, standard randomness removal techniques, 
 incur a multiplicative cost of $2^k$ in the XOR model. We show
how to reduce this factor to $O(k)$.

In addition, we prove analogous error reduction and randomness removal
results in the other models, separate all models from each other,
and show that some natural problems~-- including Set Intersection and
Find the First Difference -- separate the models when the Hamming
weights of their inputs is bounded. Finally, we show how to use the rank lower bound technique for our weak output models.

\end{abstract}

\section{Introduction}

Most of the literature on the topic of communication complexity has focused on 
Boolean functions. The usual definition 
stipulates that at the end of the protocol, one of the
players knows the value of the function.  In the rectangle based lower
bounds, the assumption is slightly stronger: at the end of the
protocol, the transcript of the protocol determines a combinatorial
rectangle of inputs that all evaluate to the same outcome. This
means that given the transcript (together with the public coins, in
the randomized public-coin setting), an external
observer can determine the output. In the case of
Boolean functions, this assumption makes no significant difference
since the player who knows the value of the function can send it in
the last message of the protocol, at an additional cost of at most one 
bit.  When the function has large outputs, however, sending the value
of the function as part of the transcript could cost more than all the
prior communication.  When this happens, then what should be
considered the ``true'' communication complexity of the problem?

When studying functions with large outputs, several fundamental questions and issues
emerge.  What lower bound techniques extend to non-Boolean functions?
When
composing protocols with large outputs, it may not be useful for
both players to know the values of the intermediate functions, and
the aggregated cost of relaying the outcome at each intermediate step
could exceed the
complexity of the composed problem. These issues are also applicable
to information complexity, where the cost is measured in information
theoretic terms instead of in number of bits of
communication. Requiring protocols to reveal the outcome as part of the transcript
could be an obstacle to finding very low information protocols.
{{It also raises the following issue: how does one amplify success when outputs are large? Amplification schemes typically involve 
repeating a protocol and taking a majority outcome, but finding said majority outcome naïvely incurs a cost that depends on the length of the output. We explore these issues, and give new models and amplification schemes.}}

Well-studied examples of functions with large outputs 
include asymmetric games, like the
NBA problem~\cite{Orlitsky90,Orlitsky91} (see also \cite[Example 4.53, p. 64]{KushilevitzN1997}), and many
problems where the output is essentially of the same size as the input
(e.g., computing the intersection of two sets~\cite{BrodyCKWY14,BCKWY_algo16}).  A
decisional analog of a function with large output may have a similar
communication complexity (e.g.,
Set Disjointness~\cite{KalyanasundaramS92,Razborov92,MR2059642-BarYossefJKS-2004})
or a very different one (e.g., deciding if the parties' numbers sum to
something greater than a given
constant~\cite{Nisan,Viola2015}).
{Large output functions also appear when studying whether multiple instances of the same function exhibit economies of scale, known as direct sum problems,
along with their variants such as agreement and elimination~\cite{aaronson2005complexity,ambainis2001communication,beimel2010choosing}.
}
In these and other problems, computing one bit of the output can be just as hard or 
significantly easier than computing the full output, 
depending on the function and on the model. 
Finally, simulation protocols, whose output are transcripts of another protocol,
have played a key role in compression \cite{BravermanR2010,Braverman2015,kol2016interactive,Sherstov-2016-compression-product,BravermanK18} as well as structural results \cite{MR3306971-reversenewman,DBLP:journals/jacm/GanorKR16,GanorKR2015,RaoS2015}.
The Find the First Difference problem has been 
 instrumental in compression
protocols. Better protocols are known when weaker output
conditions are required~\cite{BarakBCR2010,BauerMY15}.

\subsection{Output models}
We put forward several natural alternatives to the model where the
transcript {and public randomness} reveal (possibly without containing it explicitly) the
value of the function (we call this the \emph{open} model).  In the
\emph{local} model, both players can determine the value of the
function locally (but an external observer  might not be able to do
so {-- unlike in the open model}).
In the \emph{unilateral} model, one player always learns the answer.  In the 
\emph{one-out-of-two} model, the player who knows the answer can vary. 
{In the \emph{split} model, the bits of the output are split between the players in an arbitrary way known to both players.}
Finally, in the \emph{XOR model}, each player outputs a string and the
result is the bitwise XOR of these outputs. 
{The models form a hierarchy, shown in \cref{fig:hierarchy}. 
We defer formal definitions to \cref{app:outputs-models}.}

\begin{figure}[t]
  \center
  \includegraphics[page=2,width=\linewidth]{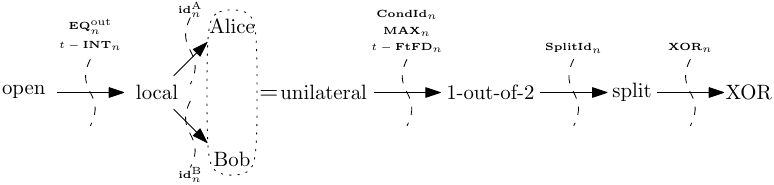}
  \caption{The various models of communication and problems separating
    them.
    An arrow
    from A to B indicates that a communication protocol for a task of
    type A is also a communication protocol for a task of type B.
    {Details of the separations are provided in \cref{app:outputs-models,app:separating-problems}.}
    }
\label{fig:hierarchy}
\end{figure}

In the context of
protocols, {we make a distinction between what the players output
and what the protocol computes. For example, in the XOR model, players output
strings $a$ and $b$ but the result of the protocol is \mbox{$a \oplus b$}.}
We will use the word ``output'' to designate what the
players output at the end of the protocol, and ``result'' or
``outcome'' to be the outcome of the protocol (which should be
 -- either probably or certainly --  the value or output
of the function).
Similarly, we will use the term
``protocol'' to designate the full mechanism for producing the result,
and ``communication protocol'' for the interactive part of the
protocol where the players exchange messages, not including the output
mechanism.

Among all the models we propose, the XOR model is perhaps the most
interesting.  This model was partly inspired by
(quantum) XOR games, 
where the players do not exchange any
messages (for example~\cite{Bell64,PalazuelosV16,BuhrmanCMW10}). 
{
One interesting property of the XOR model is that 
it could be the case, for example, that the output of each player, taken individually, follows a uniform distribution%
\footnote{Any protocol in this model can be converted into a protocol of same complexity with this property: the players  pick 
a shared random string $r$ of the same length as the output, and output $a\oplus r$ ($b\oplus r) $, where $a,b$ were the outputs of the original protocol.},   revealing nothing about either of the inputs or even the value of the function when run as a black box.}

Moreover, it is common in communication complexity to consider the
complexity of Boolean functions 
composed with some  ``gadget'' applied to the inputs. For example, for a Boolean function $f$, one can ask what is the communication complexity of $F(u,v)=f(u\oplus v)$, where bitwise XOR is applied as a gadget on the inputs.
The XOR model can be seen as applying the XOR gadget to the outputs instead of the inputs:
the players output $(a,b)$, and we require $F(u,v) = a \oplus b$ for the computation to be correct.

\subsection{Our contributions}

We focus on the XOR model where the players each output a string and the 
outcome of the protocol is the bitwise XOR of these strings.

\paragraph*{Error reduction.}
We consider the question of error reduction in \cref{sec:error-reduction}.
Error reduction is usually a simple task: repeat
a computation enough times, and take the majority outcome. 
However, in the XOR model, neither of the players knows any of the outcomes,
so neither can compute the majority outcome without additional communication.
Sending over all the outcomes so one of the players can compute the majority
would add a prohibitive  $\Theta(k)$ term, where $k$ is the length of the output.
Removing this dependency on $k$ is possible, however, and doing so requires quite
elaborate protocols that highlight the inherent limitations of the XOR model 
(\textbf{\cref{thm:error-reduction-no-k}}).

We further improve the dependency on the error parameter $\epsilon$ 
for direct sum problems (\textbf{\cref{thm:error-reduction-direct-sum}}),
by combining protocols for amortized Equality~\cite{FederKNN95} and
Find the First Difference~\cite{FeigeRPU1994}, as well as 
Gap Hamming Distance~\cite{IndykW03,ChakrabartiR12,Vidick13,Sherstov12}. 
\paragraph*{Deterministic versus randomized complexity.}
In \cref{sec:derandomization}, we revisit the classical
result that states that for any Boolean function, 
 the deterministic communication complexity is at most exponential in the private coin randomized complexity.
Once again, if the size of the output is $k$, then applying
existing schemes naively to our weaker models
adds a multiplicative cost of
$2^k$.   We show that a dependency of a factor of $k$ suffices 
(\textbf{\cref{thm:derand-xor}}). 
\paragraph*{Gap Majority composed with XOR.}
To prove our
results for the XOR model, we consider the non-Boolean \emph{Gap Majority}
problem composed with an XOR gadget. In the standard majority problem, the input is a set of elements and the goal is to find the element which appears most often.
The gap majority problem adds the promise that the majority element should appear at least some a fixed fraction (more than half) of the time. Composition with an XOR gadget turns the problem into a communication complexity problem
(see \cref{sec:error-reduction} and \cref{app:cc-hxor}).
We show that the communication
complexity of this problem is closely related to the problems of reducing error and removing randomness 
in the XOR model.  
\paragraph*{Other models and separations.} 
We define several communication models and give problems
that maximally separate them (\cref{app:outputs-models}). 
We revisit error reduction and randomness removal in other models 
(\cref{app:proofs-error-reduction,app:derandomization}).
The randomness removal scheme for the one-out-of-two model uses a variant
of the NBA problem in a subtle way as part of the reconciliation of
the majority candidates of the two players.
We reduce the dependency on~$k$ to a factor of $\log(k)$ in the one-out-of-two model, 
and remove this dependency entirely when the error parameter $\epsilon$ is bounded by $1/3$ ({\cref{thm:derand-oot}}).

Finally, we study a few additional
problems which exhibit gaps between our various communication models. In particular, several common problems exhibit a gap when the Hamming weights of their inputs are bounded
(\cref{app:separating-problems}).

\paragraph*{Rank lower bound.}
We show how lower bound techniques can be adapted to our weak output models by revisiting the notion of monochromatic rectangles associated with the leaves of a protocol tree.
We focus on the rank lower bound on deterministic communication and show that it can be used in all of our models, including the XOR model. (\cref{sec:large-output-rank})

It is important to note that our results mostly do not apply
to large-output \emph{relations} (such as the variants of direct sum, elimination and agreement), 
as many of our proofs crucially 
rely on the fact that there is a single correct answer.

\section{Related work}

{Previous works have addressed the question of the output model for large output functions.}
Braverman et al.~\cite{BravermanRWY13}
make a distinction between ``simulation'' and ``strong simulation'' of
a protocol. In a strong simulation, an external observer can determine
the result  without any knowledge of the inputs.  In
their paper on compression to internal information~\cite{BauerMY15},
Bauer et al.~stress the importance, when compressing to internal
information, that the compression itself need not reveal information
to an external observer.  They consider two output models which they
call internal and external computation.  In external computation
(which we call the open model), an external observer can determine the
result of the protocol, whereas in internal computation (which we call
the local model), the players both determine the result at the end of
the protocol.%
\footnote{ We prefer the terms \emph{open} and \emph{local} to avoid any confusion
  between the notions of \emph{internal} and \emph{external} computation, and
  \emph{internal} and \emph{external}
  information. 
}
They observe that in the deterministic setting, for total functions, the two models coincide,
but they can differ in the distributional setting.
They consider a key problem of 
finding the first bit where two strings differ, when each player has one of the two strings.
This problem is used in reconciliation protocols to find the first place where transcripts differ.
Feige et al.~\cite{FeigeRPU1994} 
externally (openly) solve Find the First Difference in $O(\log(\frac n \epsilon))$,
which was shown to be tight by Viola~\cite{Viola2015}.
Bauer et al.~\cite{BauerMY15} give an internal (local) protocol 
with a better complexity, 
where the improvement depends on the entropy of the input distribution.

\section{Preliminaries}

An introduction to  communication complexity can be found
in Kushilevitz and Nisan's~\cite{KushilevitzN1997}, and  
Rao and Yehudayoff's~\cite{RaoY18} textbooks. 

We denote by $\cX$ (resp.\ $\cY$) the set of inputs of Alice
(resp.\ Bob), $\cR_\ali$ her private randomness ($\cR_\bob$ for Bob), and
$\cR^\pub$ the public randomness accessible to both players. When
$\card{\cX}=\card{\cY}$, we  denote by $n$ the size of the input (so that $n
= \ceil{ \log(\card{\cX}) } $). When computing a function, we 
denote by $k$ the length of the output, $\cZ$ the \emph{image} of the function 
 and $k = \ceil{ \log(\card{\cZ}) } $. 
We  sometimes consider
an additional output symbol $\top$. 

We define a \emph{full protocol} as the combination of
a \emph{communication protocol} and an \emph{output mechanism} (this
 is discussed in \cref{app:outputs-models}). We
define a (two-player) communication protocol~$\Pi$ as a full binary
tree where each non-leaf node $v$ is assigned a player $\cP^v$ amongst
$A$(lice) and $B$(ob), and a mapping $\cN^v$ into $\zo$ whose input
space depends on which player the node was assigned to. When $\cP^v=A$
(resp.\ $B$) then $\cN^v$'s input space is $\cX
\times \cR_\ali \times \cR^\pub$ (resp.\ $\cY \times \cR_\bob \times \cR^\pub$). Note
that the tree and each node's owner are fixed and do not depend on the
input.
In an execution of a communication protocol,
the two players walk down the tree together, starting from the root,
until they reach a leaf. Each step down the tree is done by letting
the player who owns the current node $v$ apply its corresponding
mapping $\cN^v$, and sending the result to the other player. If it is
$0$, the players replace the current node by its left child, and
otherwise by its right child.  The \emph{communication cost} $\CC(\Pi)$ of a
protocol $\Pi$ is the total number of bits exchanged for the worst case
inputs.

Since an execution of a communication protocol $\Pi$ is entirely
defined by the players' inputs ($(x,y) \in \cX \times \cY$) and the
randomness (the players' private randomness $r_\ali \in \cR_\ali$ and $r_\bob \in \cR_\bob$ as
well as the public randomness $r \in \cR^\pub$), we also view the
communication protocol as a function $\Pi : \cX \times \cY \times
\cR_\ali \times \cR_\bob \times \cR^\pub \rightarrow \zo^*$ whose values we
call \emph{transcripts} of $\Pi$.  {For the purposes of this paper, we do not include} the public randomness as part of the transcript.
For a given protocol $\Pi$, we
denote by $T_\pi = \Pi(X,Y,R_\ali,R_\bob,R)$ the random variable over transcripts
of the protocol that naturally arises from $X$, $Y$, $R_\ali$, $R_\bob$, and
$R$, taken as random variables. We denote by $\cT_\pi$ the support of
the distribution $T_\pi$. 
 We denote by $x , y , z , r_\ali , r_\bob , r , t_\pi$ elements of the
  sets $\cX , \cY , \cZ , \cR_\ali , \cR_\bob , \cR^\pub , \cT_\pi$, respectively,
  which in turn are the supports of the random variables $X, Y, Z, R_\ali, R_\bob, R, T_\pi$.

We recall definitions and known bounds of functions that will be
used in this paper. For all of these problems,
note that the communication complexity is of the same order of
magnitude whether we require that both players
know the output or only one of them, since the size of the output is
no larger than the communication required for one player to know the
output. 
{In the remainder of this section, we denote by $R_\epsilon(f)$ the
minimal communication cost of a randomized protocol computing function
$f$ with error at most $\epsilon$ when, say, Bob outputs. $D(f) =
R_0(f)$ denotes the deterministic communication complexity.}

\begin{definition}[Find the First Difference problem]
  $\FtFD_n : \zo^n \times \zo^n  \rightarrow  \set{0,\ldots,n}$ is defined as
$
    \FtFD_n (x,y)
        =  \min (\set{i : x_i \neq y_i} \cup \set{n}).
$
\end{definition}

\begin{proposition}
\label{prop:ftfd}
For any $0<\epsilon<\frac 1 2$,
$R_\epsilon(\FtFD_n) \in \Theta(\log( n ) + \log(1/\epsilon))$~\cite{FeigeRPU1994,Viola2015}.
\end{proposition}

The upper bound uses a walk on a tree 
where steps are taken according to results from hash functions. The
lower bound is from a lower bound on the Greater Than function
$\GT_n$, which reduces to $\FtFD_n$. For a good exposition of the
upper bound, see Appendix~C in~\cite{BarakBCR2010}.

\begin{definition}[Gap Hamming Distance problem]
\label{def:ghd}
  Let $n,L,U$ be integers such that $0 \leq L < U \leq n$.
    $\GHD_n^{L,U} : \zo^n \times \zo^n  \rightarrow  \zo$
is a promise problem where the input satisfies the promise that the
Hamming distance between inputs $x,y$ is either $\geq U$ or $\leq L$. Then $\GHD_n^{L,U} (x,y) = 1$ in the first case and $0$ in the second case.
\end{definition}

The bounds on Gap Hamming Distance vary  depending on the parameters. 
In this paper we use a linear upper bound 
which is essentially tight in the regime we require.
Many other bounds are known for other regimes~\cite{Kozachinskiy15,ChakrabartiR12,Vidick13,Sherstov12,BuhrmanCW98,Watson18}.

\begin{definition}[Equality problem]
\label{def:eq}
    The function  $\EQ_n : \zo^n \times \zo^n  \rightarrow  \zo$ is defined as 
$\EQ_n(x,y)={\indic}_{x = y}$.
The $k$-fold Equality problem is $  \EQ_n^{\otimes k} ( (x_1,\ldots,x_k) , $ $(y_1,\ldots,y_k) )  =  $ $ (\EQ_n(x_1,y_1),\ldots,\EQ_n(x_k,y_k))$, where $(x_i,y_i) \in \zo^n$ for all~$i$. 
\end{definition} 
\begin{proposition}
\label{prop:eq}
For $0<\epsilon < \frac 1 2$,
$R_\epsilon(\EQ^{\otimes k}_n) \in \Theta(k + \log(1/\epsilon))$.
\end{proposition}

The algorithm from~\cite{HPZZ_siamcomp21} which achieves optimal communication uses hashing just like the algorithm for a single
instance. It saves on communication compared to $k$ successive uses of a protocol for equality with error $\epsilon/k$ by having players hash all $k$ instances simultaneously, exchange results, and repeat this process, exploiting that they have less and less to communicate about. Intuitively, the number of unequal instances to discover should decrease as the algorithm runs.
Once it has been determined for an instance $(x_i,y_i)$ that $x_i \neq y_i$ through unequal hashes, the players do not need to speak further about this instance. An unequal instance is unlikely to survive many tests, which means that late in the algorithm the players can exchange their hashes using that most of them should agree. The idea was also present in previous algorithms~\cite{FederKNN95} which improved on the trivial algorithm.
The lower
bound is just from $\Omega(k)$ bits of communication being necessary to
send $k$ bits worth of information, even with~$\epsilon$ error.

Unless otherwise specified, our protocols use both private and public coins. We 
use the `$\priv$' superscript when 
the protocols and mappings do not have access to 
public randomness. 

\section{The \emph{XOR} model}

In the XOR model, each player outputs a string and the value of the
function is the bitwise XOR of the two outputs
(\cref{model:xor}).
This model is inspired by XOR games which have been widely studied
in the context of quantum nonlocality as well as unique games.

\begin{restatable}[XOR computation]{definition}{xormodel}
\label{model:xor}
  Consider a function $f$ whose output set is $\cZ = \zo^k$. A
  protocol $\Pi$ is said to  \emph{XOR}-
  compute $f$ with $\epsilon$ error if there exist two mappings
  $\cO_\ali$ and $\cO_\bob$ with $\cO_\ali : \cT_\pi \times \cR^\pub \times
  \cR_\ali \times \cX \rightarrow \zo^k$ and similarly $\cO_\bob : \cT_\pi \times \cR^\pub \times
  \cR_\bob \times \cY \rightarrow \zo^k$
  such that   for all $(x,y) \in \cX \times \cY$, 
  $$\Pr_{r,r_\ali,r_\bob}[ \cO_\ali(t_\pi,r,r_\ali,x) \oplus \cO_\bob(t_\pi,r,r_\bob,y) = f(x,y) ] \geq 1-\epsilon.$$
\end{restatable}

We define $D^\xor(f)$ (resp.~$R_\epsilon^\xor(f)$)
as the best communication
cost of any protocol that computes $f$ in the XOR model with
error~$\epsilon = 0$ (resp.\ with error at most $\epsilon$, for $0 <
\epsilon < \frac 1 2$).   
(Notations are defined similarly for our other models with superscripts
$\op, \loc, \ali, \bob, \uni, \spt, \oot$.)

\section{Error reduction and the Gap Majority problem}
\label{sec:error-reduction}

We study the cost of {reducing the error} of communication
protocols in our weaker models of communication where the outcome of
the protocol is not known 
to both of the
players.  We focus on the more interesting case of the XOR model in
the main text, and results for the other models are in 
\cref{app:error-reduction-up-to-XOR}.

{Standard error reduction schemes work by repeating a protocol
many times in order to compute and output the most frequently occurring value  among
all the executions. Repeating the protocol enough times ensures that with
high probability, the output that appears the most is correct.}
One can derive an upper bound on the number of iterations needed
from Hoeffding's inequality.
\begin{lemma}[Hoeffding's inequality]
  \label{lem:hoeffding}
  Let $\parens*{ V_i }_{i\in[N]}$ be $N$ {independent} Bernouilli
  trials of expected value~$p$. We have
  $\Pr\event*{ \abs[\big]{\tfrac{1}{N} \sum_{i=1}^N V_i - p } \geq \delta }
  \leq 2\cdot \exp\parens*{-\frac {\delta^2 N}{2p(1-p)}}.$
\end{lemma}

The following  holds in the  setting where Bob outputs the 
value of the function at the end of the protocol.
\begin{theorem} (Folklore, see~\cite{KushilevitzN1997})
  \label{thm:amp-kn97}
  Let $0<\epsilon'<\epsilon<\frac 1 2$, and
  $C_{\epsilon,\epsilon'} = \frac {2 \epsilon (1 - \epsilon)}{\parens* {
    \frac 1 2 - \epsilon }^2} \ln \parens*{ \frac{2} {\epsilon'}
  }$.  For all functions $f : \cX \times \cY
  \rightarrow \cZ$,  
  $R^\bob_{\epsilon'}(f) \leq C_{\epsilon,\epsilon'} \cdot R^\bob_\epsilon(f).$
\end{theorem}

Note that it is important here that $f$ is a function, not a relation,
so that there is a unique correct output and the player(s) can compute
the majority.

In the XOR model, finding the majority result among some number $T$ of
runs is much more difficult than in the standard model, since neither of
the players can identify reasonable candidates as the majority answer.
Exchanging all of the $T$ $k$-bit outputs would result in a bound of
$R_{\epsilon'}^\xor(f) \leq C_{\epsilon,\epsilon'} \parens*{
R^\xor_\epsilon(f) +
k %
}.
$
We show that this dependence on $k$ is unnecessary. 

\begin{theorem}\label{thm:error-reduction-no-k}
  Let $0 < \epsilon' < \epsilon < \frac 1 2$, 
  $C_{\epsilon,\epsilon'} =    8\epsilon {\parens*{ \frac 1 2 - \epsilon }^{-2}}{\ln \parens*{\frac 8 {\epsilon'}}}$.  For all $f : \cX \times \cY \rightarrow \zo^k$,
 $ R^{\xor}_{\epsilon'}(f) \leq 
   C_{\epsilon, \epsilon'} \cdot R^\xor_\epsilon(f) 
   + O \parens*{
              C_{\epsilon, \epsilon'} 
  }\ .
$
\end{theorem}

In order to prove this result, we introduce the Gap Majority ($\GapMAJ$) problem, 
show how \cref{thm:error-reduction-no-k} reduces to solving $\GapMajX$
(\cref{lem:amp-xor-hxor}), then give an upper bound on solving
$\GapMajX$ (\cref{thm:hxor-no-k}).

The partial function $\GapMAJ[N,k,\epsilon,\mu]$ has $N$ strings of length $k$ as input and 
the promise is that there is a string $z$ of length $k$ that appears with $\mu$ weight at least $(1-\epsilon)$ among the $N$ strings, where $\mu$ is a  distribution over indices in $[N]$.

\begin{definition}[Gap Majority]
\label{def:hxor}
In the Gap Majority problem  $\GapMAJ[N,k,\epsilon,\mu]:\parens*{ \zo^k }^N\rightarrow \zo^k $ 
the input is $(Z_1,\ldots,Z_N)$,
 and 
$\mu$ is a fixed distribution over the indices $[N]$.
  When unspecified, 
  $\mu$ is understood to be the
  uniform distribution.
  The promise is that $\exists z \in
  \zo^k$ such that $\mu\parens*{\set*{i \in [N] : Z_i  =
  z}} \geq (1-\epsilon)$. Then
  \[  \GapMAJ[N,k,\epsilon,\mu] 
                         (( Z_i )_{i\in[N]}) 
            =  z \quad \textrm{s.t.} \quad \mu\parens*{
               \set*{ i : Z_i = z }}
            \geq (1-\epsilon) .
   \]
\end{definition}

In $\GapMajX$, the players are given $N$ strings of length $k$
and their goal is to compute $\GapMAJ$ on the bitwise XOR of their inputs whenever the $\GapMAJ$ promise is satisfied. 
(Notice that when $k=1$, 
this is equivalent to the Gap Hamming Distance problem~(\cref{def:ghd})
with parameters {$L=\epsilon N$, $U=(1-\epsilon)N$}.) 

For inputs $(X_1,\ldots,X_N),(Y_1,\ldots,Y_N)$ to  $\GapMajX[N,k,\epsilon,\mu]$,
we will refer to a pair $(X_i, Y_i)$ as a \emph{row}, and we call $X_i$ Alice's $i$th row, and $Y_i$ Bob's $i$th row.
As a warm-up exercise, we show that  error reduction
reduces to solving an instance of $\GapMajX$.
\begin{lemma}\label{lem:amp-xor-hxor}
  Let $0 < \epsilon' < \epsilon < \frac 1 2$ and $C_{\epsilon,\epsilon'} =   2 \epsilon {\parens*{\frac 1 2 - \epsilon}^{-2}} { \ln \parens*{\frac 4 {\epsilon'} }}$. 
For every $f : \cX \times \cY \rightarrow
  \zo^k$,
$
  R_{\epsilon'}^\xor(f) 
	\leq 
	C_{\epsilon, \epsilon'}\cdot R_\epsilon^\xor(f) 
	+ R_{\epsilon' / 2}^\xor\parens*{
 \GapMajX[C_{\epsilon,\epsilon'},k,\frac 1 4 + \frac {\epsilon} 2]
 }.
$
\end{lemma}

\begin{proof}[Proof of \cref{lem:amp-xor-hxor}]
  Let $\pi$ be  a protocol which XOR-computes $f(x,y)$ with $\epsilon$-error and $\pi'$ be a protocol which computes $\GapMajX[C_{\epsilon,\epsilon'},k,\frac 1 4 + \frac {\epsilon} 2]$ in the XOR model, with error $\epsilon'/2$.
We consider the following protocol, which we denote by $\hat \pi$: first, run $\pi$ $C_{\epsilon,\epsilon'}$ times; then, use the outputs produced by this computation as inputs for $\pi'$, run the latter protocol, and output the result.
     We analyze the new protocol $\hat \pi$ as follows.
     The outputs produced in the first step 
     are strings $X_1,\cdots,X_{C_{\epsilon, \epsilon'}}$ on Alice's
side, and  $Y_1,\cdots,Y_{C_{\epsilon, \epsilon'}}$ for Bob. A run of $\pi$ is correct iff $X_i \oplus Y_i=f(x,y)$.
By  Hoeffding's bound (\cref{lem:hoeffding}), applied with $N = C_{\epsilon,\epsilon'}$, $V_i = 1$ if $X_i \oplus Y_i \neq f(x,y)$ and $V_i = 0$ otherwise  for $i=1,\dots,N$, $p = \Exp[V_i] \leq \epsilon$, and $\delta = \frac{1}{2}(\frac{1}{2} - \epsilon)$, we get that with probability at least $1-2e^{-\delta^2N / (2p(1-p))} \geq 1 - \epsilon'/2$, a fraction $p + \delta \leq (\frac{1}{2} + \epsilon)/2$ of the $N$ computations err.
In other words, with probability at most $\epsilon'/2$, the above strings fail to satisfy the promise in the definition of $\GapMajX[C_{\epsilon,\epsilon'},k,\frac 1 4 + \frac {\epsilon} 2]$. {Conditioned} on this not happening (i.e., on the promise being met), $\pi'$ (hence~$\hat\pi$) errs with probability at most $\epsilon'/2$. The overall error is  at most $\epsilon'$.
\end{proof}

To derive a general upper bound on error reduction using
\cref{lem:amp-xor-hxor}, it would suffice to have
an  upper bound on $R^\xor_{\epsilon'}(\GapMajX[N,k,\epsilon])$.
When the error parameter is large ($\epsilon\leq \epsilon'$), $\GapMajX$ in the XOR model is trivial: the
players just need to sample a common row and output according to that
row.  However, \cref{lem:amp-xor-hxor} requires solving a $\GapMajX$ 
instance with small error~$\epsilon'/2$, which takes us back to square one:
finding an error reduction scheme that we can apply to $\GapMajX$.

In the remainder of the section, we give a protocol for $\GapMajX$ 
(\cref{sec:hxor-amp})
followed by an error reduction scheme
for direct sum functions (\cref{sec:direct-sum}). In both cases, we use the structure of the
XOR function and a protocol for Equality on pairs of rows to find
a majority outcome.
The error reduction scheme for direct sum functions is a refinement of
\cref{lem:amp-xor-hxor} and is useful in cases where the
starting error is very close to $\frac 1 2$ and where computing one bit
of the output is significantly less costly than computing the full
output.

\subsection{Solving \texorpdfstring{$\GapMajX$}{HXOR}}
\label{sec:hxor-amp}
Given an instance of $\GapMajX[N,k,\epsilon]$, if Alice and Bob pick a
row and output what they have on this row, they get the correct output
with probability $\geq 1 - \epsilon$. 
Recall that we would like to achieve error $\epsilon' < \epsilon$ without
incurring a dependence on parameter $k$, which in our application to error reduction corresponds to the length
of the output. We show that this is possible.

\begin{theorem}\label{thm:hxor-no-k}
  Let $0 < \epsilon' < \epsilon < \frac 1 2$,
  $R^{\xor}_{\epsilon'}(\GapMajX[N,k,\epsilon]) \leq  
    O \parens*{
    N + \log \parens*{\tfrac 1 {\epsilon'}}
    }\ .$
\end{theorem}

\idea
We use the fact that $a\oplus b
= a'\oplus b'$ iff $a\oplus a' = b\oplus b'$. Therefore, the players
can identify rows that XOR to a same string by solving instances of
Equality. This idea alone is enough to obtain a protocol
for $\GapMajX[N,k,\epsilon]$ of complexity
$O\parens*{ N^2 + \log \parens*{ \tfrac 1 {\epsilon'} } }$ by
computing Equality for all $\binom{N}{2}$ pairs of rows to identify the majority outcome.
We improve on this by reducing
the number of computed Equality instances using Erd\H{o}s-R\'enyi random graphs
(\cref{lem:connected-component}).

\begin{lemma}[Variation of eq.\ $(9.18)$ in~\cite{ErdosR60}]\label{lem:connected-component}
  Let $G(n,p(n))$ be the distribution over graphs of $n$ vertices where each edge is sampled with independent probability
  $p(n)$. Let $L_1(G)$ be the size of the largest connected component
  of $G$. Then: 
  \[\forall \alpha \in [0,1], c \in \bbR^+,
  \qquad\Pr\range*{L_1(G(n,c/n)) < (1-\alpha)n}
  \leq e^{\parens*{
  \ln(2) - \frac \alpha 2 \parens*{1-\frac \alpha 2} c
  } n} .\]
  In particular this probability goes to $0$ as $n$ goes to infinity when $\alpha c > 4 \ln(2)$.
\end{lemma}

For completeness, the  proof is given in \cref{app:proofs-error-reduction}.

\begin{proof}[Proof of \cref{thm:hxor-no-k}]
 Consider the $\GapMajX$ instance as a $N\times k$ matrix such that
 $(X_i)_{i \in [N]}$ are the rows of Alice and $(Y_i)_{i \in [N]}$ are
 the rows of Bob.  By the promise of the $\GapMajX$ problem, we know
 there exists a  $z\in\zo^k$ such that $\set{i:X_i \oplus Y_i =z}
 \geq (1-\epsilon)N$. The goal is now for Alice and Bob to identify a
 row belonging to this large set of rows that XOR to the same $k$-bit
 string.

Let $i$ and $j$ be the indices of two rows.
The event that the two rows XOR
to the same string is expressed as $X_i \oplus Y_i = X_j \oplus Y_j$,
which is equivalent to $X_i \oplus X_j = Y_i \oplus Y_j$. This means
that we can test whether any two rows XOR to the same bit string with
a protocol for Equality.

The protocol goes through the following steps:
\begin{enumerate}
  \item The players pick rows randomly, enough rows so that with high
    probability, a constant fraction of the rows XOR to the {majority element}
    $z$.
  \item The players solve instances of Equality to find large sets of rows
    that XOR to the same string. In each such large set of rows, they pick a single row. This leaves them with a constant
    number of candidate rows that might XOR to the {majority element} $z$.
  \item The players decide between those candidates by comparing them
    with all the rows. There is one candidate row that XORs to the same string
    as most rows; this row XORs to the {majority element} $z$.
\end{enumerate}

\begin{description}
  \item[\hypertarget{step:many-iterations-f}{Step 1.}] 
  Using public randomness, Alice and Bob now pick a
    multiset $S$ of all their rows of size $\card{S} =
    T_{\epsilon'}=50\ln\parens*{ \frac {10} {\epsilon'}}$.  Each
    element of $S$ is picked uniformly and independently.  Using
    Hoeffding's inequality (\cref{lem:hoeffding}), with
    probability $\geq 1 - \frac {\epsilon'} 5$ more than $\frac 2 5$
    of those executions XOR to the {majority element} $z$. 
  \item[\hypertarget{step:first-equality-batch}{Step 2.}]
  We now consider $S$ as the vertices $V$ of a random
    graph $G=G(V,E)$, in which each edge is picked with a probability
    $\frac c {\card{V}}$ with $c>0$. Consider the subgraph $G'$ of $G$
    induced on the vertices $V' \subseteq V$ that correspond to
    executions that XOR to the {majority element} $z$.  From the previous step, we
    know that $\card{V'} \geq \frac 2 5 T_{\epsilon'} = 20 \ln\parens*{ \frac
    {10} {\epsilon'}}$. The subgraph $G'$ is a random graph
    where each edge was picked with the same probability $\frac c
    {\card{V}} = \frac {c'} {\card{V'}}$ where $c' = c \frac {\card{V'}}{\card{V}} \geq
    \frac 2 5 c$. By \cref{lem:connected-component}, this
    subgraph $G'$ contains a connected component of size $\geq
    (1-\frac 1 {12})\card{V'} \geq \frac {11} {30}\card{V}$ with probability
    $\geq 1 - 2^{-\card{V'}} \geq 1-\frac {\epsilon'} 5$ for $c \geq \frac
    {720}{143}\ln(2) \approx 3.49$ as $ \card{V'} \geq 20 \ln\parens*{ \frac
    {10} {\epsilon'}} \geq
    \log{\parens*{\frac{5}{\epsilon'}}}$.

At this point, Alice (resp.\ Bob) computes the bitwise XOR of all pairs
of executions that correspond to an edge in $G$: $(X_i \oplus
X_j)_{(i,j)\in E,i<j}$ (resp.\ $(Y_i \oplus Y_j)_{(i,j) \in
  E,i<j}$). For $\epsilon'$ small enough, with high probability ($\geq
1 - \frac {\epsilon'} 5$), the set of edges of $G$ is smaller than
$2c\cdot T_{\epsilon'}$ by Hoeffding's inequality (the players can
abort the protocol otherwise). Then, Alice and Bob solve $\leq 2c\cdot
T_{\epsilon'}$ instances of Equality with (total) error $\leq \frac
{\epsilon'} 5$ to discover a large set of rows that XOR to a same bit
string. We now have groups of rows that we know XOR to the same bit
string, at least one of which represents more than $\frac {11} {30}$
of $S$'s rows because of the Hoeffding argument combined with the
random graph lemma.

Now for each submultiset of rows of $S$ that XOR
to the same bit string and represents more than $\frac {11} {30}$ of all of
$S$'s rows, pick an arbitrary row in the 
  submultiset. If there is only one such 
  submultiset, Alice and Bob can end the protocol here, 
  outputing the content of the
 row selected in this submultiset. If there were two
such submultisets, then let $i_1$ and $i_2$
be the indices picked in each submultiset.

\item[Step 3.]
To decide between their two candidates, Alice and Bob solve $N$
Equality instances between $X_{i_1} {\oplus} X_j$ and $Y_{i_1} {\oplus}
Y_j$ for all $j\in \range{N}$ with error~$\leq \frac {\epsilon'}
5$. If more than half of the $N$ rows XOR to the same string as the
$i_1^{th}$ row, Alice and Bob output their $i_1^{th}$ row. Otherwise,
they output the other candidate row $i_2$.

\end{description}

The complexity of computing $\GapMajX[N,k,\epsilon]$ with error $\epsilon'<\epsilon$ satisfies
  \label{eqn:amp-xor-no-k}
  \begin{align*}
  R^{\xor}_{\epsilon'}\parens*{ \GapMajX[N,k,\epsilon] }
  & \leq
    R_{\epsilon'/5}\parens*{ \EQ_k^{\otimes 2 c T_{\epsilon'}} }
    + R_{\epsilon'/5}\parens*{ \EQ_k^{N} }\ .
  \end{align*}

To conclude, we apply an amortized protocol for 
Equality (\cref{prop:eq}).
\end{proof}

Combining \cref{lem:amp-xor-hxor} and \cref{thm:hxor-no-k}
concludes the proof of \cref{thm:error-reduction-no-k}. 
We will return to the  $\GapMajX$ problem
in 
\cref{app:cc-hxor}
where we give upper bounds in various models (\cref{cor:simple-hxor-bounds}).

\subsection{XOR Error reduction for direct sum functions}
\label{sec:direct-sum}

The protocol of  \cref{thm:error-reduction-no-k} 
first generates a full instance of $\GapMajX$, then solves this instance. 
The generation of this instance might create an implicit dependency on the output length $k$ of $f$, which
in the regime where $\epsilon$ is very close to $1/2$ can be prohibitive. 
We give a different protocol in which the players are not required to fully generate these intermediate results.

For large output functions, generating one bit of the output can be much less {costly} than generating all $k$,
for example, when $f$ is a direct sum of $k$ instances of a function
$g$. We state our stronger amplification theorem for the case of direct sum problems of Boolean functions, but
we note that the protocol could be used for other problems where computing one bit of the output
is less costly than computing the entire output.

\begin{theorem}\label{thm:error-reduction-direct-sum}
  Let $0 < \epsilon' < \epsilon < \frac 1 2$ and 
$C_{\epsilon,\epsilon'}=   8 \epsilon {\parens*{ \frac 1 2 - \epsilon } ^{-2}}{\ln \parens*{\frac {12} {\epsilon'}}}.$
For any $g : \cX \times \cY \rightarrow \zo$ 
and $f=g^{\otimes k}$,
\label{eq:amp-xor-version}
 \begin{align*}
  R^{\xor}_{\epsilon'}(f) & \leq 
   50 \ln \parens*{ \tfrac {12} {\epsilon'} }\cdot R^\xor_\epsilon(f) 
  +C_{\epsilon, \epsilon'} \cdot R^\xor_\epsilon(g) 
  + O \parens*{
	 C_{\epsilon, \epsilon'}
  	+ \log (k) 
  }\ .
  \end{align*}

\end{theorem}
Notice that {the $C_{\epsilon, \epsilon'}$ factor -- which scales with $\parens*{\frac{1}{2} -\epsilon}^{-1}$ --} applies to the complexity of $g$, not of $f$. 

\idea
Instead of iterating the basic protocol {$C_{\epsilon,\epsilon'}$} times, we will start by iterating it a smaller number of times
which does not depend on $\epsilon$, but only on $\log(\frac {1}{\epsilon'})$.
{This number of iterations suffices to guarantee that the most frequent outcome represents more than a $1/3$ fraction of the rows. If no other outcome represents a large fraction of the rows, we output according to a row from this large fraction. Otherwise, still, at most two outcomes can represent more than a $1/3$ fraction of the rows.} We identify a ``critical index'' of the output function, one that will help 
us {identify the majority result among the two candidate outcomes.}
We do so by solving a Gap Hamming Distance instance on the critical index.
In these remaining {$C_{\epsilon,\epsilon'}$} runs, we only need one of the $k$ bits of the output.

Details of the proof are given  in 
\cref{app:direct-sum}.

\section{Deterministic versus randomized complexity}
\label{sec:derandomization}

We now turn to removing randomness from private coin protocols.

{The standard scheme to derive a deterministic protocol from a private coin protocol\footnote{For public coins, the exponential upper bounds do not hold, for example in the
case of the Equality function, which has an $O(1)$ public coin randomized protocol,
but requires $n$ bits of communication to solve deterministically.} proceeds as follows~\cite[Lemma 3.8, page~31]{KushilevitzN1997}.  
The players  exchange messages to estimate the probability of each transcript.
They use the fact that the probability of a transcript  can be factored into two parts,  each of which  can be computed by one of the two players. One of the players sends all of its factors to the other, up to some precision, and the second player can then estimate the probability of each transcript. Each transcript determines an output,  therefore from the estimate for the transcripts' probabilities, this player can derive an estimate for the probability of each output, and output the majority answer.
}

\begin{theorem}[Lemma 3.8 in~\cite{KushilevitzN1997}, page~31]
For any function $f:\cX\times\cY \to\cZ$ and $0<\epsilon < \frac 1 2$, let $R =
R^{\priv}_\epsilon(f)$. Then 
  \label{thm:derand-kn97}
$D(f) \leq 2^{R}
	\parens*{ 
	R+
	\log \parens*{\tfrac {1}{\frac 1 2 - \epsilon} } + 1
	}.$
\end{theorem}
Using this well-known result for our output models (first adding $k$ bits of
communication to the original protocol of cost $R$ to obtain a
protocol that works in the unilateral model) would add $2^{R}R\cdot
2^k$ bits to the complexity.  For the XOR model, we reduce the
dependency to a $O(2^Rk)$ term.  In 
\cref{app:derandomization},
we show some lower
dependencies on~$k$ in our other models.

We formalize the problem which we call Transcript Distribution
Estimation. Let $\Delta(\mu,\nu)=\frac 1 2 \sum_{u \in \cU} \abs{\mu(u) - \nu(u)}$ be the
total variation distance between two probability distributions $\mu$
and $\nu$ over a universe $\cU$. For a protocol $\Pi$, let $\cT_\pi$ be the set of transcripts of $\Pi$, and for $(x,y) \in \cX \times \cY$, let us denote by
$T^{x,y}_\pi$ the distribution over $\cT_\pi$ witnessed when running $\Pi$ on $(x,y)$.

The key step of the proof of 
\cref{thm:derand-kn97} is a protocol (in the standard model)
for the following problem.
 
  \begin{definition}[Transcript Distribution Estimation problem]
For any protocol $\Pi$ and $\delta<\frac 1 2$, we say that a protocol $\widetilde \Pi$ solves $\TDE_{\Pi,\delta}$ in model $\cM$ if, for each input $(x,y)$, $\widetilde \Pi$ computes in the sense of model $\cM$ 
a distribution $\widetilde T^{x,y}_\pi$ such that $\Delta(\widetilde T^{x,y}_\pi,T^{x,y}_\pi) \leq \delta$.
\end{definition}

\begin{lemma}[Implicit in~\cite{KushilevitzN1997}, page~31]
    \label{lem:tde-kn97}
  Let $\Pi$ be a private coin communication protocol and $\cT_\pi$ its set of possible transcripts. 
For any $0<\delta< \frac 1 2$, 
  $ D(\TDE_{\Pi,\delta}) \leq \abs{\cT_\pi} \cdot \ceil*{
  \log \parens*{\frac {\card{\cT_\pi}}{\delta} }
  }.$
\end{lemma}

In their proof, Kushilevitz and Nisan~\cite{KushilevitzN1997} require only one of the players to learn an estimate of the 
probability of each leaf. Here we require both players 
to learn
the same estimate, which can be achieved with a factor of two in the communication.
Details are given in 
\cref{lem:tde-open} in \cref{app:tde}. 

In the XOR model, however, sharing such an estimate is not sufficient to  remove randomness.
At each leaf, each player outputs values with some probability (depending on their
private randomness), so there can be as many as $\card{\cZ}$  outputs per leaf by each player, 
making identifying the majority outcome impossible.
We prove the following bound on deterministic communication in the XOR model.

\begin{theorem}
        \label{thm:derand-xor}
   Let $0<\epsilon <1/2$ and $f : \cX \times \cY \rightarrow \cZ = \zo^k$. Let $R=R^{\xor,\priv}_\epsilon(f) $, $M=16 \cdot 
   \parens*{\frac 1 2 - \epsilon }^{-2} 
   \cdot 2^{R}$, 
   and $\epsilon' = \frac{5}{8} - \frac{\epsilon}{4}$.
   Then 
\begin{align*}
	D^\xor(f) &\leq   D(\TDE_{\Pi_f,\epsilon'-\frac 1 2}) 
 + D^\xor(\GapMajX[M,k,\epsilon',\mu])
\\ & \leq
\parens*{ 2^{R+1} } 
 \cdot \parens*{R + \log \parens*{\tfrac 8 {\frac 1 2 - \epsilon }} + 1}
 + k \cdot \parens*{\frac {5-2\epsilon} 4 M+1} .
\end{align*}
Where $\mu$ is an unspecified distribution over $[M]$.
\end{theorem}
\idea
We reduce the problem of finding the majority 
outcome to a much smaller instance of $\GapMajX$ by discretizing the probabilities of the outputs. This 
 lets us  reduce the dependence on the size of the output
to just a factor of $k = \log(\card{\cZ})$ (instead of a factor of $2^{2k} = \card{\cZ}^2$).

\begin{proof}[Proof of \cref{thm:derand-xor}]
 Let $\Pi$ be an optimal private coin XOR protocol for $f$. The players start
 running the $\TDE_{\Pi,\delta}$ protocol of
 \cref{lem:tde-open} 
 (\cref{lem:tde-kn97} adapted to the local model, see \cref{app:tde} for details) 
with $\delta = \frac 1 4 \parens*{ \frac 1 2
 - \epsilon }$, thus learning within statistical distance
 $\delta$ the probability distribution over
 leaves that results from the protocol.

 Let $o_\ali(. \mid w,x)$ and $o_\bob(. \mid w,y)$ be the two independent probability
 distributions over $\zo^k$ according to which Alice and Bob output,
 conditioned on reaching leaf $w$, having received inputs $x$ and $y$.
To reduce the problem to  $\GapMajX$,
 they discretize $o_\ali$ and $o_\bob$ into $\ceil{\delta^{-1}}$
 events. Let $\dot o_\ali$ denote the discretization of $o_\ali$ with
 following properties for Alice (Similarly for $\dot o_\bob$ for  Bob):
 \[
 \forall z,w:\qquad \dot o_\ali(z \mid w,x) \cdot \ceil{\delta^{-1}} \in \bbN
\quad\text{and}\quad
\abs{o_\ali(z \mid w,x) - \dot o_\ali(z \mid w,x)}
               \leq \frac 1 {\ceil{\delta^{-1}}}.
\]
A simple greedy approach to discretization goes like this:
\begin{enumerate}
\item Replace all $o_\ali(z \mid w,x)$ by $\dot o_\ali(z \mid w,x) =
  \frac 1 {\ceil{\delta^{-1}}} \floor*{ {\ceil{
    \delta^{-1}}} o_\ali(z \mid w,x) }$.
\item While the probabilities of $\dot o_\ali$ sum to less than $1$,
  pick a $z$ s.t.\ $o_\ali(z \mid w,x) - \dot o_\ali(z \mid w,x)$ is maximal. For
  that $z$, set $\dot o_\ali(z \mid w,x) = \frac 1 {\ceil{\delta^{-1}}}
  \ceil*{ {\ceil{\delta^{-1}}} o_\ali(z \mid w,x) }$.
\end{enumerate}

The players then construct a distributional $\GapMajX$ instance with $M$ rows
where $M = \ceil{\delta^{-1}}^2 \card{\cT_\pi}$ in the following way:

\begin{itemize}
\item For each leaf $w$ the players define $\ceil{\delta^{-1}}^2$ rows.
      Rows are indexed by $(i,j) \in \range*{ \ceil{\delta^{-1}} }
      \times \range*{ \ceil{\delta^{-1}} }$ and are such that:
      \begin{itemize}
      \item For each $z$, there are exactly $\ceil*{ \delta^{-1}
            } \dot o_\ali(z \mid w,x)$ indices
            $i_z \in \range*{ \ceil{\delta^{-1}} }$ such that
            Alice outputs $z$ on all rows of the form $(i_z,j), \forall j$.
      \item For each $z$, there are exactly $\ceil{ \delta^{-1}
            } \dot o_\bob(z \mid w,y)$ indices
            $j_z \in \range*{ \ceil{\delta^{-1}} }$ such that
            Bob outputs $z$ on all rows of the form $(i,j_z), \forall i$.
      \end{itemize}
\item The probability of the row $(i,j)$ associated to the leaf $w$
  under the distribution $\mu$ is taken to be $p^\lf(w \mid x,y) \cdot
  \ceil{\delta^{-1}}^{-2}$,  where $p^\lf(w \mid x,y)$
    is the probability of ending in a leaf $w$ in the original
    protocol $\Pi$. ($\mu$ is the unspecified distribution
    over $[M]$ in the statement of \cref{thm:derand-xor}.)
\end{itemize}

The players then solve the $\GapMajX$ instance and output the result. Clearly, the above procedure  has the
previously claimed communication complexity. It remains to show that
the players  built a valid $\GapMajX$ instance whose result is
$f(x,y)$, that is,  picking a random row according
to $\mu$ from this $\GapMajX$ instance gives outputs $z_\ali$ and $z_\bob$ on
Alice and Bob's sides such that $z_\ali \oplus z_\bob = f(x,y)$ with
probability $> \frac 1 2$.

\begin{enumerate}
        \item  In the original protocol $\Pi$, 
          let $p^\out(z \mid x,y)$ be the probability of computing
          $z$ (after the XOR), $p^\out(z \mid w,x,y)$
          that same probability conditioned on the protocol ending in
          leaf $w$, and for all $w$ let $o_\ali(. \mid w,x)$ (resp.\ $o_\bob(. \mid w,y)$)
          be the  distribution according to which Alice (resp.\ Bob) outputs
          once in leaf $w$.  Then
            $p^\out(z \mid x,y)$ can be expressed as:
\begin{eqnarray*}
            p^\out(z \mid x,y) &=& \sum_w p^\lf(w \mid x,y) \cdot p^\out(z \mid w,x,y)\\
            &=& \sum_w p^\lf(w \mid x,y) \cdot \sum_{\mathclap{\substack{z_\ali,
              z_\bob\\z_\ali \oplus z_\bob = z}}} o_\ali(z_\ali \mid w,x)\cdot
            o_\bob(z_\bob \mid w,x).
\end{eqnarray*}

            By correctness of the protocol, $p^\out(f(x,y) \mid x,y) \geq 1 -
            \epsilon$. 
        \item Consider $p'^\lf(. \mid x,y)$, $p'^\out(. \mid x,y)$,
          $p'^\out(. \mid w,x,y)$, $\dot o_\ali(. \mid w,x)$ and $\dot o_\bob(. \mid w,y)$
          the approximations of the above quantities 
          encountered when building our instance of $\GapMajX$. The
          probability $p'^\out(z \mid x,y)$ that a random row of our weighted
          $\GapMajX$ instance corresponds to a given $z$ is:
          \[
          p'^\out(z\mid x,y) = \sum_w p'^\lf(w \mid x,y) \cdot
          \sum_{\mathclap{\substack{z_\ali, z_\bob\\z_\ali \oplus z_\bob = z}}} \dot
          o_\ali(z_\ali \mid w,x)\cdot \dot o_\bob(z_\bob \mid w,x) .
          \]
        \item $p'^\lf(. \mid x,y)$ is $\delta$-close to $p^\lf(. \mid x,y)$ in
          statistical distance. $\dot o_\ali(. \mid w,x)$ is point-wise
          $\delta$-close to $o_\ali(. \mid w,x)$ (and similarly for $\dot o_\bob$
          and $o_\bob$).  \label[point]{enum:delta-close}
\end{enumerate}

Consider $o_\ali\cdot o_\bob$ the distribution over $z\in\zo^k$
 defined by $o_\ali\cdot o_\bob(z)=\sum_{z'}o_\ali(z' \mid w,x)\cdot o_\bob(z\oplus
 z' \mid w,y) $. Similarly define $ o_\ali \cdot \dot o_\bob$ and $\dot
 o_\ali \cdot \dot o_\bob$. \Cref{enum:delta-close} above implies that
 $\dot o_\ali \cdot \dot o_\bob$ is point-wise $\delta$-close to
 $o_\ali \cdot \dot o_\bob$, which is itself point-wise $\delta$-close to
 $o_\ali \cdot o_\bob$. One can check that  $\dot o_\ali \cdot \dot o_\bob$ is point-wise
 $2\delta$-close to $o_\ali \cdot o_\bob$. 

 Using \cref{lem:leafs-and-outputs} (in the appendix)
 with $V
  \sim p^\out$, $V'\sim p'^\out$, $U \sim p^\lf$, $U'\sim
  p'^\lf$, $V_u \sim o_\ali \cdot o_\bob$ and $V'_{u} \sim \dot o_\ali \cdot
  \dot o_\bob$, we get that $p$ and $p'$ are point-wise $3\delta$-close.
Since $\delta$ was taken to be $\frac 1 4 \parens*{\frac 1 2
- \epsilon }$, the probability that the random row of the $\GapMajX$ instance corresponds to $f(x,y)$ is:
$
p'^\out(f(x,y)) \geq p^\out(f(x,y))-3\delta \geq (1 - \epsilon)
     - \frac 3 4 \parens*{\frac 1 2 - \epsilon } 
     = \frac 1 2 + \frac 1 4 \parens*{\frac 1 2 - \epsilon } > \frac 1 2.
 $
\end{proof}


\section{Rank lower bounds for weak output models }
\label{sec:large-output-rank}

  Since the output requirements are weaker in our new models, standard lower bound techniques may
 no longer apply. 
 We adapt the
  standard rank lower bound to all of
  our output models (\cref{thm:rank-lb-large-output}). While we
  do not prove any new lower bound with this result, the main contribution of this section is to show how to adapt an existing
  lower bound to our new communication complexity models. Our techniques can also be applied to  other lower bound techniques in a similar fashion.

\paragraph{Reconsidering monochromatic rectangles}

Let $f:\cX\times\cY\rightarrow \bbF$ where the value of the function is interpreted as an element of a field $\bbF$. The communication matrix associated with $f$ is the matrix whose rows are indexed by elements of $\cX$ and columns by elements of $\cY$ and is defined as $M_f = (f(x,y))_{x\in\cX, y\in \cY}$.

In the open model, since there is a mapping from leaf nodes to
outputs, a communication protocol partitions the communication matrix into
monochromatic rectangles.
This is not the case with the other models of computation.
In the unilateral and one-out-of-two models, the
rectangles at the leaves are ``striped'' horizontally or vertically (see \cref{fig:stripes} for an illustration), since a player can
change her answer depending on her input. In the unilateral models,
the direction of the stripes is always the same in all rectangles,
while the stripes  can have different directions  in the one-out-of-two model, depending on which player produces the output.
The
local model is more subtle: the two players can decide to output different elements
of $\bbF$ depending of their local information (their input and
randomness). Whenever the two players output something different, the
result is incorrect, which gives their rectangles a look similar to
permutation matrices. 

In the rest of this section we will 
use the term ``leaf rectangle'' to designate rectangles corresponding to leaves  of the protocol tree.

\begin{figure}[ht]
  \center
  \begin{subfigure}[t]{0.3\textwidth}
   \center
   \includegraphics[page=1,scale=0.7]{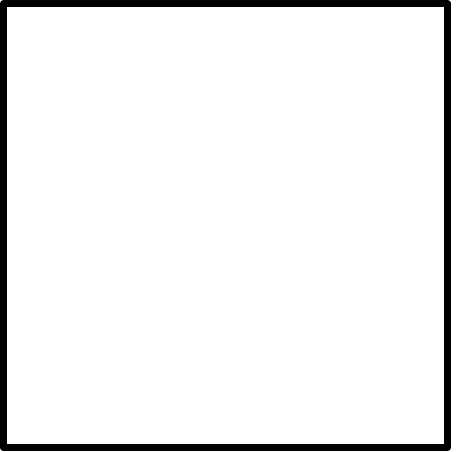}
  \end{subfigure}
  ~
  \begin{subfigure}[t]{0.3\textwidth}
   \center
   \includegraphics[page=2,scale=0.7]{rank_matrices}
  \end{subfigure}
  ~
  \begin{subfigure}[t]{0.3\textwidth}
   \center
   \includegraphics[page=3,scale=0.7]{rank_matrices}
  \end{subfigure}
  \caption{Rectangles corresponding to leaves at the end of a protocol in the open model are monochromatic, while in
    the unilateral and the one-out-of-two models they have
    monochromatic ``stripes'', by which we mean that they are further  partitioned into monochromatic subrectangles by partitioning the rows (for horizontal stripes) or the columns (for vertical stripes) according to what is output by the player responsible for outputting the function value. The direction (horizontal or vertical) depends on which player outputs the value of the function.}
\label{fig:stripes}
\end{figure}

The situation of the split and the XOR models is somewhat different,
as their leaf rectangles have a more complicated structure. In the XOR model,
the leaf rectangles generated by a XOR protocol are similar to the
communication matrix for the XOR function $\XOR_k$.

\paragraph{Rank lower bound}

In order to derive rank lower bounds for our models, we study the ranks of the leaf rectangles. The ranks of the leaf rectangles for the various models imply the following theorem.

  \begin{theorem}
    \label{thm:rank-lb-large-output}
  Let $f$ be a total function. Then
  \begin{align*}
  D^\op(f) &= D^\loc(f) \geq D^\uni(f) \geq D^\oot(f) \geq \log \rank(M_f)\\
 D^\spt(f) &\geq \log \rank(M_f) - 1\\
  D^\xor(f) &\geq \log \rank(M_f) - \log(k + 1)
  \end{align*}
\end{theorem}

\begin{proof}
  Let us call rank of a rectangle of $M_f$ the rank of the submatrix
  of $M_f$ obtained by restricting $M_f$ to the rectangle. If there
  exists a partition of $M_f$ into $C$ rectangles such that the rank
  of each rectangle is bounded by $R$, then $\rank(M_f) \leq C \times
  R$. Since for every model $\cM$, $M_f$ is covered by at most
  $2^{D^\cM(f)}$ rectangles of type $\cM$, we only need to bound
  the rank of rectangles of type $\cM$ for each model $\cM$.
  
  \paragraph*{Open, local, unilateral, and one-out-of-two leaf rectangles.}
    Leaf rectangles of these types are of rank at most $1$, because of
    their striped structure. Also note that open and local leaf rectangles
    are similar for total functions in the deterministic setting.

  \paragraph*{Split leaf rectangles.} Leaf rectangles of this type are of rank at most
    $2$. Intuitively, this is because the leaf 
    rectangles in this model are of the following form: there exists
    numbers $a_1,\ldots, a_s$ and $b_1, \ldots, b_t$ such that the
    value of the cell $(i,j)$ of the rectangle of size $s \times t$,
    is $a_i + b_j$. The rectangle is then the product of the following
    two rank-$2$ matrices: the $s \times 2$ matrix containing the
    values $a_1$ to $a_s$ in the first column and the value $1$ in all
    cells of the second column and the $2 \times t$ matrix containing
    only the value $1$ in its first line and the values $b_1$ to $b_t$
    in the second line, as shown in \cref{fig:split-rectangles}.

    \begin{figure}[ht]
      \center
      \includegraphics[page=1,scale=1]{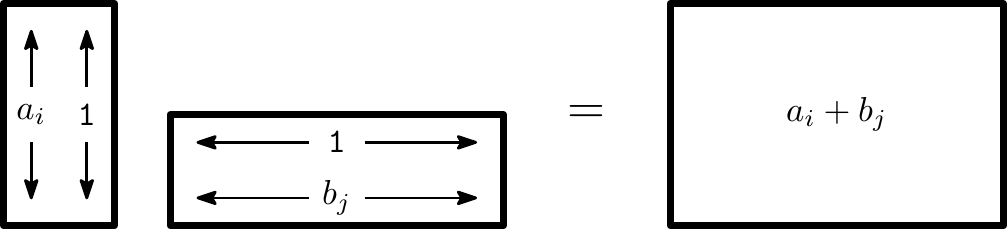}

      \caption{A matrix whose cells $M_{i,j}$ can be expressed as the
        sum of the $i$th entry of a first vector and the $j$th entry
        of another one is of rank at most 2. Split rectangles follow
        this pattern.}
      \label{fig:split-rectangles}
\end{figure}

    More formally: consider how the $k$ bits of the output are
    split between the two players: let us consider the $k$ bit string
    $(s_i)_{1\leq i \leq k}$ such $s_i = 1$ iff Alice outputs the
    $i^{th}$ bit of the output.

    Let us now define the $1 \times 1$ matrix $S_0 = \begin{bmatrix}
      0\end{bmatrix}$, and let $H_c$ and $V_c$ the matrix
      transformations defined by:
      \begin{itemize}
      \item $H_c(A) = \begin{bmatrix} A & A+c\cdot J\end{bmatrix}$
      \item $V_c(A) = \begin{bmatrix} A \\ A+c\cdot J\end{bmatrix}$
      \end{itemize}

      Now we define three series of matrices $S_1\ldots S_k$,
      $U_0\ldots U_k$ and $V_0\ldots V_k$ such that $S_i = U_i \times
      V_i$ for all $i$, which will prove that $S_k$ has rank at most $2$:

      \begin{itemize}
      \item Let $S_{i+1} =
      \begin{cases}H_{2^i} (S_i) & \textrm{if $s_i=0$} \\
        V_{2^i} (S_i) & \textrm{if $s_i=1$}
      \end{cases}$.
      \item Let $U_0 = \begin{bmatrix} 0 & 1\end{bmatrix}$ and $V_0
        = \begin{bmatrix} 1 \\ 0\end{bmatrix}$.
      \item Let $U_{i+1} =
          \begin{cases} U_i & \textrm{if $s_i=0$} \\
            \begin{bmatrix}U_i \\ U_i +
              \begin{bmatrix} 2^i & 0 \\ \vdots & \vdots \\ 2^i & 0 \end{bmatrix} \end{bmatrix} & \textrm{if $s_i=1$}
          \end{cases}$
        \item Let $V_{i+1} =
          \begin{cases}
            \begin{bmatrix}V_i & V_i +
              \begin{bmatrix} 0 & \ldots & 0 \\2^i & \ldots & 2^i \end{bmatrix} \end{bmatrix} & \textrm{if $s_i=0$} \\
            V_i & \textrm{if $s_i=1$} \\
          \end{cases}$
      \end{itemize}

  To see that the property $S_i = U_i \times V_i$ is true for all $i
  \in [k]$, notice that the second column of $U_i$ and the top row of
  $V_i$ only contain $1$'s, since this is true for $i=0$ and the
  property is preserved as $i$ increases. Adding a constant $c$ to the
  second half of the second line of $V_i$, this constant gets
  multiplied by the second column of $U_i$, that only contains
  $1$'s. The end result is that we add a $c \cdot J$ matrix to half of
  the matrix, which is exactly what we want.

  Finally, notice that $S_k$ is a matrix containing all that Alice and
  Bob can output in the split model given a specific split. A
  leaf rectangle in the split model is a submatrix of a matrix of this
  form, where some lines and columns have possibly been permuted or
  duplicated. Therefore, leaf rectangles in the split model have rank at
  most $2$.
      
  \paragraph*{XOR leaf rectangles} We prove that leaf rectangles produced by XOR
    protocols have rank at most $(k+1)$.

    Consider the communication matrix of the $\XOR_k$ function. An XOR
    leaf rectangle can be obtained as a submatrix of this communication
    matrix, possibly after permuting or duplicating some rows and
    columns. Thus, it suffices to show that $M_{\XOR_k}$ has rank
    $k+1$. We do this by directly giving a rank $k+1$ decomposition of
    $M_{\XOR_k}$.
    Consider the following $2^k \times 1$ vectors:
    \begin{itemize}
    \item $v^k$ is the all-one vector.
    \item For $0 \leq i < k$, $u^{k,i}$ is such that $u^{k,i}_j =
      (-1)^{1+j_i}$ (for $0\leq j < 2^k$). Such vectors are sometimes
      called Hadamard vectors.
    \end{itemize}
    
    Let $S_k$ be the following $2^k \times (k+1)$ matrix:
    
    \[S_k = \begin{bmatrix} \sqrt{2^{k-1}-2^{-1}}\cdot v & \ {\sqrt{2^{-1}}} \cdot
      u^{k,0}& \ldots & {\sqrt{2^{k-2}}} \cdot u^{k,k-1}
      \end{bmatrix}\]

    \begin{figure}[ht]
  \center
  \begin{subfigure}[t]{0.23\textwidth}
   \center
   \includegraphics[page=1,scale=0.6]{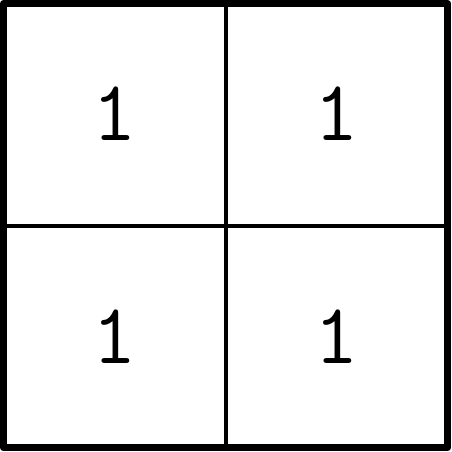}
   \caption{}
  \end{subfigure}
  ~
  \begin{subfigure}[t]{0.23\textwidth}
   \center
   \includegraphics[page=2,scale=0.6]{xor_decomp}
   \caption{}
  \end{subfigure}
  ~
  \begin{subfigure}[t]{0.23\textwidth}
   \center
   \includegraphics[page=3,scale=0.6]{xor_decomp}
   \caption{}
  \end{subfigure}
  ~
  \begin{subfigure}[t]{0.23\textwidth}
   \center
   \includegraphics[page=4,scale=0.6]{xor_decomp}
   \caption{}
  \end{subfigure}

  \caption{The $M_{\XOR_3}$ communication matrix can be obtained by a
    linear combination of those matrices.}
  \label{fig:xor-decomposition}
\end{figure}

    We have that $S_k \transpose{S_k} =
    M_{\XOR_k}$. \Cref{fig:xor-decomposition} gives an intuition
    of how the $M_{\XOR_k}$ matrix is obtained.
\end{proof}

\section{Conclusion and open questions}
\label{sec:open-questions}

We have presented output models that 
are tailored for  non-Boolean functions. 
We hope that these will 
find many applications, including extensions to information complexity,
a better understanding of direct sum problems, simulation protocols,
new lower bounds tailored to these models, to name just a few.

The Gap Majority composed with XOR problem (\cref{def:hxor})
is closely related
to the Gap Hamming Distance, extended to 
a large alphabet but with an additional promise, so 
lower bounds for $\GHD$ do not apply. 
We conjecture that its deterministic communication complexity is 
$\Omega(\epsilon N k)$, matching the trivial  upper bound. 
If true, this would indicate that our
randomness removal scheme (\cref{thm:derand-xor}) is close to
tight.

\section{Acknowledgments}
\label{sec:discussion}

We thank Jérémie Roland and Sagnik Mukhopadhyay for helpful
conversations and the anonymous referees for their numerous suggestions to improve the paper's presentation.
This work was funded in part by the ANR grant FLITTLA ANR-21-CE48-0023.

\pagebreak

\bibliography{0_new_refs_2022}

\pagebreak
\appendix

\section{Models for large-output functions}
\label{app:outputs-models}

One standard definition of communication complexity requires that at the end of the communication protocol, the output of the
computation can be determined from the transcript of the
communication and the public randomness (it is the model used in rectangle bounds).
It is  easy to find examples
where such a definition  makes it necessary to exchange much more
communication than seems natural. For example,
\begin{example}
  \label{ex:promise-gap}
  Consider the function $f:\{0,1\}^n \times \{0,1\}^n \rightarrow
  \{0,1\}^n, f(x,y) = x$, and assume we want to compute it with the
  promise $x=y$.
\end{example}

A protocol for $f$ requires $n$ bits of communication if the result of
the protocol has to be apparent from the communication and the public
randomness, even though both players know $f(x,y)$ right from the
start.

  In this section, we formally define the output models and prove
  separation results. The most interesting models are
  arguably the weakest ones: the one-out-of-two
  (\cref{model:oot}), the split
  (\cref{model:split}), and the XOR models
  (\cref{model:xor}).

\subsection{The open model}
\label{sec:open-model}

We start with the formal definition of our model which  reveals the most information regarding the outcome of the computation. We call it the \emph{open} model. 

This is the model
for which the partition bounds~\cite{JainK10}, in the form in which they appear in the literature, give
lower bounds.

\begin{definition}[Open computation]
\label{model:open}
  A protocol $\Pi$ is said to \emph{openly} compute $f$ with $\epsilon$ error
  if there exists a mapping $\cO : \cT_\pi \times \cR^\pub \rightarrow
  \cZ$ such that:  for all $(x,y) \in \cX \times \cY$, 
  \[\Pr_{r,r_\ali,r_\bob}[ \cO(t_\pi,r) = f(x,y)] \geq 1-\epsilon.\]
\end{definition}

\subsection{The local model}

In the previous model, protocols are \emph{revealing}, in the sense
that the result of the computation can not be a secret only known to
the players. In the \emph{local} model, we only require that both
players, at the end of the protocol, can output the value of the
function (or the same valid output, in the case of a relation).


\begin{definition}[Local computation]
\label{model:local}
  A protocol $\Pi$ is said to \emph{locally} compute $f$ with
  $\epsilon$ error if there exist two mappings $\cO_\ali$ and $\cO_\bob$
  with $\cO_\ali : \cT_\pi \times \cR^\pub \times \cR_\ali \times \cX
  \rightarrow \cZ$ and similarly $\cO_\bob : \cT_\pi \times \cR^\pub \times \cR_\bob \times \cY
  \rightarrow \cZ$ such that:  for all $(x,y) \in \cX \times \cY$, 
  \[\Pr_{r,r_\ali,r_\bob}[ \cO_\ali(t_\pi,r,r_\ali,x) = \cO_\bob(t_\pi,r,r_\bob,y) = f(x,y)] \geq 1-\epsilon. \]
\end{definition}

Bauer et al.~\cite{BauerMY15} remarked that for total functions
and relations, the deterministic open and local communication
complexities are the same. \Cref{ex:promise-gap} shows a
separation between the deterministic complexities of computing a
function with a promise.

For randomized communication, the local model is separated from the
open model by the following total function, as seen in
\cref{thm:eqout-open-loc-sep}:

\begin{definition}[Equality with output problem]
\label{def:eqout}
$\EQout_n : \{0,1\}^n \times \{0,1\}^n  \rightarrow  \{0,1\}^n \cup \{\top\} $ 
is defined as 
\[
  \EQout_n (x,y) =  
  	\begin{cases} 
		x	& \text{if $x=y$}\\ 
		\top 	& \text{otherwise} 
	\end{cases}
\]
\end{definition}
\begin{figure}[ht]
  \center
  \includegraphics[page=4]{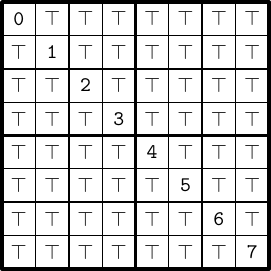}
  \label{fig:eqout}
  \caption{The communication matrix of $\EQout_3$}
\end{figure}


\begin{theorem}
  \label{thm:eqout-open-loc-sep}
  $\forall f : \cX \times \cY \rightarrow \cZ$ with $k=  \ceil*{ \log \card{\cZ} } $ and $\epsilon > 0$,
\begin{gather*}
R_\epsilon^\loc(f)  \leq R_\epsilon^\op(f) \leq R_\epsilon^\loc(f) +k, \quad \text{ and}\\
 R^{\loc}_{1/4} (\EQout_n)  \leq 4, \quad \quad R^{\op}_{1/4} (\EQout_n)  \in \Omega(n).
\end{gather*}
\end{theorem}

We provide a full proof of this theorem, but because all the results
of the form $R_\epsilon^{\cM_1}(f) \leq R_\epsilon^{\cM_2}(f)$ or
$R_\epsilon^{\cM_1}(f) \leq R_\epsilon^{\cM_2}(f) +k$ for two models
$\cM_1$ and $\cM_2$ can be proved by essentially the same proof, we
will omit them in proofs of later similar theorems, only proving the
separation result.

\begin{proof}[Proof of \cref{thm:eqout-open-loc-sep}]

  \begin{description}
    \item[Proof of $R_\epsilon^\loc(f) \leq R_\epsilon^\op(f)$:] An
      open protocol for a function $f$ is also a local protocol for
      $f$, as the players can take as mappings $\cO_\ali$ and $\cO_\bob$ the
      mapping $\cO$ of the open protocol (ignoring both players'
      randomness and input).
    \item[Proof of $R_\epsilon^\op(f) \leq R_\epsilon^\loc(f) +k$:]
      Let $\Pi$ be a local protocol for computing $f$ with error at
      most~$\epsilon$. Consider $\Pi'$, the protocol that consists of
      first running the protocol $\Pi$, and then Alice sends
      $\cO_\ali(t_\pi,r,r_\ali,x)$ -- what she would output at the end of
      $\Pi$ to locally compute $f$ -- over the communication
      channel. This only requires $k$ additional bits of
      communication. Now $\Pi'$ is an open protocol, since an
      external observer can use the last $k$ bits of the transcript as
      probable $f(x,y)$.
  \end{description}
  
  Both the lower bound and the upper bound on $\EQout$ directly follow
  from propositions and theorems previously seen in this manuscript.
  
  \begin{description}
    \item[Local model upper bound:] The players apply the standard
      protocol for $\EQ$ (\cref{prop:eq}). If the strings
      are different, they output $\top$, otherwise Alice outputs $x$
      and Bob outputs $y$.
    \item[Open model lower bound:] Consider the
  mapping $\cO$ of the open protocol $\Pi$ and notice that for all
  $x$, $\Pr_r[\cO(\Pi(x,x,r),r)=x] \geq 3/4$. Consider that
  the players have a public randomness source $\cR^\pub$ that is the
  uniformly random distribution over $\zo^k$. Then the above
  statement implies $\card{\cO^{-1}(x)} \geq \frac 3 4 \cdot 2^k$. Since
  $\cup_x \cO^{-1}(x) \subseteq \cT_\pi \times \zo^k$, we have that
  $\frac 3 4 \cdot 2^k \cdot 2^n \leq 2^{\CC(\Pi)} \cdot 2^k$ hence $\CC(\Pi) \geq n
  + \log \parens*{\frac 3 4} \in \Omega(n)$. This is also
  true when the source of public randomness is not a uniform
  distribution over $\set{0,1}^k$ because of the fact that any
  non-uniform source of randomness can be simulated with arbitrary
  precision by a uniform source of randomness.  
  \end{description}
\end{proof}

In \cref{app:wprt} we generalize this to show that any open protocol for a problem requires $\Omega(k)$ communication. This result follows from a lower bound known as the weak partition bound~\cite{FontesJKLLR16}.

\subsection{The unilateral models}

In this section, we consider models of communication complexity where
we require that at the end of the protocol, one player can output the
value of the function (or a valid output, in the case of a
relation). One-way problems are usually stated in this model.


\begin{definition}[Unilateral computation]
\label{model:uni}
  A protocol $\Pi$ is said to \emph{Alice}-compute $f$ with
  $\epsilon$ error if there exists a mapping $\cO_\ali : \cT_\pi \times
  \cR^\pub \times \cR_\ali \times \cX \rightarrow \cZ$ such that:   for all $(x,y) \in \cX \times \cY$, 
  \[\Pr_{r,r_\ali,r_\bob}[ \cO_\ali(t_\pi,r,r_\ali,x) = f(x,y)] \geq 1-\epsilon. \]

  Bob-computation is defined in a similar manner.

  A protocol is said to \emph{unilaterally} compute $f$ if it
  Alice-computes or Bob-computes $f$.
\end{definition}

Our definition of the unilateral model corresponds to a minimum of two
models, each assigned to a player.

\begin{definition}[Unilateral identity problems]
\label{def:id}
    $\id^\ali_n : \{0,1\}^n \times \{0,1\}^n \rightarrow \{0,1\}^n$ is
defined as \[\id^\ali_n (x,y) = x\]

$id^\bob_n$ is defined similarly, with opposite roles for Alice and Bob.
\end{definition}

\begin{figure}[ht]
  \center
  \begin{subfigure}[t]{0.4\textwidth}
   \center
   \includegraphics[page=7]{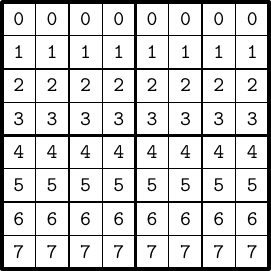}
  \end{subfigure}
  ~
  \begin{subfigure}[t]{0.4\textwidth}
   \center
   \includegraphics[page=8]{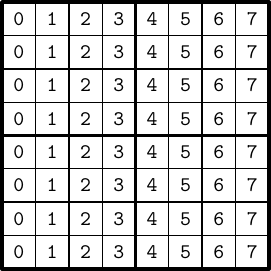}
  \end{subfigure}
  \label{fig:idA-idB}
  \caption{The communication matrix of $\id^\ali_3$ and $\id^\bob_3$}
\end{figure}

\begin{theorem}
  \label{thm:id-loc-uni-sep}
  $\forall f : \cX \times \cY \rightarrow \cZ$ with $k=  \ceil*{ \log \card{\cZ} } $, $\lambda \in [0,1]$ and $\epsilon > 0$ 
\begin{gather*}
  R_\epsilon^\uni(f) \leq  R_\epsilon^\loc(f)  \leq  R_\epsilon^\op(f)  \leq R_\epsilon^\uni(f) + k,\\
  D^{\loc}(f) \leq D^{\ali}(f)+D^{\bob}(f), \quad \quad
  R_\epsilon^{\loc}(f) \leq R_{\lambda \epsilon}^{\ali}(f)+R_{(1-\lambda) \epsilon}^{\bob}(f), \quad \text{and}\\
    D^{\uni} (\id^\ali_n) = D^{\ali} (\id^\ali_n) = D^{\bob} (\id^\bob_n)  = 0, \quad \quad
    R^{\loc}_{1/4} (\id^\ali_n) = R^{\loc}_{1/4} (\id^\bob_n)  \in \Omega(n).
\end{gather*}
\end{theorem}

The first line also holds for relations, but the second line does not:
consider as counterexample the relation $f : \{0,1\}^n \times
\{0,1\}^n \rightarrow 2^{\{0,1\}^n}, f(x,y) = \{x,y\}$. This problem
does not require any communication in both unilateral models ($D^{\ali}(f) = D^{\bob}(f) = 0$), but in
the local model, the fact that the players need to agree on a single
output makes the communication of order $\Omega(n)$ in both the
deterministic and the randomized setting ($D^{\loc}(f) \geq R_{\epsilon}^{\loc}(f) \in \Omega(n)$).

\begin{proof}[Proof of \cref{thm:id-loc-uni-sep}]  
  We omit the proof of the first two lines, that are only based on
  using the same protocol with the different proper mappings, or
  sending what one would output in a lower model over the
  communication channel.

  We prove a slightly stronger result for the separation: that
  $R^{\bob}_{1/4} (\id^\ali_n) \in \Omega(n)$.

  \begin{description}
    \item[Alice model upper bound:] Alice outputs her $x$, which
      requires no communication.
    \item[Bob model lower bound:] Let us consider
      $D^{\bob}_{1/4}(\id^\ali_n,\mu)$ where $\mu$ is the uniform
      distribution. Bob has to output one of $2^n$ equiprobable
      answers. With communication $C$, Bob can only have $2^C$
      different answers, so Bob is wrong with probability $\geq
      1-2^{C-n}$. Since Bob is supposed to make less than $\frac 1 4$
      error, we have: $C \geq n + \log \parens*{\frac 3 4}$, so $R^{\bob}_{1/4}
      (\id^\ali_n) \in \Omega(n)$.
  \end{description}
\end{proof}

\subsection{The one-out-of-two model}

In the unilateral models, the player that outputs the result at the
end of the protocol is fixed. In particular, it does not depend on the
inputs. In the one-out-of-two model, we relax this condition:
correctly computing a function in the one-out-of-two model corresponds
to an execution such that at the end of the protocol:
\begin{itemize}
\item one player outputs a special symbol $\top \not \in \cZ$ (which
  corresponds to silence)
\item the other players outputs $f(x,y)$.
\end{itemize}

Intuitively, we not only require that one of the players outputs the
correct answer, but also that she knows that her output is probably
correct, while the other knows that other player has a good answer to
output. If we were only requiring that one player gives the correct
answer, then all Boolean functions would be solved with zero
communication in this model. In contrast, our model does not
trivialize the communication complexity of Boolean functions.


\begin{definition}[One-out-of-two computation]
\label{model:oot}
  A protocol $\Pi$ is said to \emph{one-out-of-two} compute $f$ with
  $\epsilon$ error if there exist two mappings $\cO_\ali$ and $\cO_\bob$
  with $\cO_\ali : \cT_\pi \times \cR^\pub \times \cR_\ali \times \cX
  \rightarrow \cZ \cup \{\top\}$ and similarly $\cO_\bob : \cT_\pi \times \cR^\pub \times \cR_\bob \times \cY
  \rightarrow \cZ \cup \{\top\}$ 
  such that:   for all $(x,y) \in \cX \times \cY$, 
  \[
	\Pr_{r,r_\ali,r_\bob}[ \parens*{\cO_\ali(t_\pi,r,r_\ali,x),\cO_\bob(t_\pi,r,r_\bob,y)} \in \set{ (f(x,y),\top),(\top,f(x,y)) } ] \geq 1-\epsilon.
  \]
\end{definition}

The next proposition shows that any one-out-of-two protocol can be
transformed into another one-out-of-two protocol of lesser  or
equal error and using only one additional bit of communication, such
that at the end of the protocol it is always the case that exactly one
player outputs a value in $\cZ$ and the other stays silent (outputs~$\top$).

\begin{proposition}
  \label{prop:oot-one-speaker}
  Consider a function $f : \cX \times \cY \rightarrow \cZ$ and
  $\Pi$ a one-out-of-two protocol for $f$ with error $\epsilon > 0$ of
  communication cost $C$. Then there exists a one-out-of-two protocol
  $\Pi'$ of communication cost $(C+1)$ that computes $f$ with the same error but
  with mappings such that it is always the case that only one of them
  speaks at the end:
  \begin{gather*}
    \forall x,y,r_\ali,r_\bob,r,t_{\pi'}=\Pi'(x,y,r_\ali,r_\bob,r): \\
    \parens*{ \cO'_\ali(t_{\pi'},r,r_\ali,x),\cO'_\bob(t_{\pi'},r,r_\bob,y) }
  \in \parens*{ \cZ \times \{\top\} }
  \cup \parens*{ \{\top\} \times \cZ }. \end{gather*}
\end{proposition}

\begin{proof}[Proof of \cref{prop:oot-one-speaker}]
  Let $\Pi$ be a one-out-of-two protocol for $f$ and $\cO_\ali$,$\cO_\bob$
  the associated mappings. We define the protocol $\Pi'$ to be a
  protocol that first behaves as $\Pi$ (getting a transcript $t_\pi$)
  and when we hit a leaf in the protocol for $\Pi$, Alice sends a bit
  of communication to Bob following this rule:
  
  \begin{itemize}
  \item If $\cO_\ali(t_\pi,r,r_\ali,x)=\top$, Alice sends $0$ to Bob.
  \item Otherwise Alice sends $1$ to Bob.
  \end{itemize}
  Let $c_\ali$ be this control bit, sent by Alice in the last round of the new protocol $\Pi'$. Then, Alice keeps the same mapping $\cO_\ali$ whereas Bob's new mapping
  $\cO'_\bob$ is such that:

  \[\cO'_\bob(t_{\pi'},r,r_\bob,y) =
  \begin{cases} \top & \textrm{if $c_a = 1$}, \\
  \cO_\bob(t_\pi,r,r_\bob,y)  & \textrm{if $c_a = 0$ and  $\cO_\bob(t_\pi,r,r_\bob,y) \neq \top$,}    \\
  z & \textrm{picked u.a.r.\ in $\cZ$, otherwise.}
  \end{cases}\]
  
  Intuitively, Alice tells Bob whether to speak or not, and he
  obeys. Since the only cases where this changes what the players
  output is when they were going to both speak or both stay silent,
  the error does not increase in the process.
\end{proof}

\begin{definition}[Conditional identity problem]
\label{def:condid}
    The function $\CondId_n  : \{0,1\}^{n} \times \{0,1\}^{n}  \rightarrow  \{0,1\}^{n}$ is defined as  
\[
  \CondId_n(x,y) =  
		\begin{cases} 
			x & \text{if}\ x_0 = y_0,\\ y & \text{otherwise,} 
		\end{cases}
\]
where $x_0$ is the fist bit of $x$, similarly for $y$.
\end{definition}

\begin{figure}[ht]
  \center
  \includegraphics[page=9]{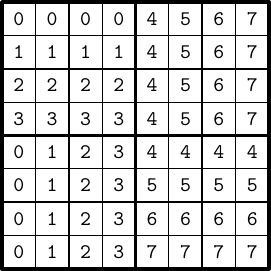}
  \label{fig:condid}
  \caption{The communication matrix of $\CondId_3$}
\end{figure}

\begin{theorem}
  \label{thm:condid-uni-oot-sep}
$\forall f : \cX \times \cY \rightarrow \cZ$ with $k=  \ceil*{ \log \card{\cZ} } $ and $\epsilon > 0$ 
\begin{gather*}
R_\epsilon^\oot(f) \leq  R_\epsilon^\uni(f) \leq  R_\epsilon^\loc(f)  \leq R_\epsilon^\op(f)  \leq R_\epsilon^\oot(f) + k + 1, \quad \text{and}
\\
    D^\oot(\CondId_n)  \in O(1),  \quad  \quad R^\uni_\epsilon(\CondId_n)  \in \Omega(n).
\end{gather*}
\end{theorem}

\begin{proof}[Proof of \cref{thm:condid-uni-oot-sep}]
  Again, we focus on the separation result.
  
  \begin{description}
    \item[One-out-of-two model upper bound:] Alice and Bob send each
      other $x_0$ and $y_0$. If $x_0=y_0$, Alice outputs $x$,
      otherwise Bob outputs $y$. This only takes $2$ bits of
      communication.
    \item[Unilateral model lower bound:] Let us consider
      $D^{\bob}_{1/4}(\CondId_n,\mu)$ where $\mu$ is the uniform distribution
      over $(x,y)$ such that $x_0=y_0$. Having received any given $x$,
      Bob has to output one of $2^{n-1}$ equiprobable answers. With
      communication $C$, Bob can only have $2^C$ different answers, so
      Bob is wrong with probability $\geq 1-2^{C-n+1}$. Since Bob is
      supposed to make less than $\frac 1 4$ error, we have: $C \geq n
      -1 + \log \parens*{\frac 3 4}$, so $R^{\bob}_{1/4} (\CondId_n) \in
      \Omega(n)$. By symmetry, we also have $R^{\ali}_{1/4}
      (\CondId_n) \in \Omega(n)$, so $R^{\uni}_{1/4} (\CondId_n) \in
      \Omega(n)$.
  \end{description}
\end{proof}

\subsection{The split model}

In our next model, we allow the answer to be split between the two
players. In the one-out-of-two model, one of the player had to output
the full output, while the other stayed fully silent. In contrast, in
the split model we allow both players to output part of the result. We
only require that any given bit is output by exactly one player (the
other player stays silent on this particular bit). In a valid split
computation, it may be that the first bit of $f(x,y)$ is output by
Alice, while the second one is output by Bob.


\begin{definition}[Split computation]
\label{model:split}
  A protocol $\Pi$ is said to \emph{split} compute $f$ with
  $\epsilon$ error if there exist two mappings $\cO_\ali$ and $\cO_\bob$
  with $\cO_\ali : \cT_\pi \times \cR^\pub \times \cR_\ali \times \cX
  \rightarrow \{0,1,*\}$ and similarly $\cO_\bob : \cT_\pi \times \cR^\pub \times \cR_\bob \times \cY
  \rightarrow \{0,1,*\}$ 
  such that:   for all $(x,y) \in \cX \times \cY$, 
  \[
	\Pr_{r,r_\ali,r_\bob}[ \cO_\ali(t_\pi,r,r_\ali,x) \weave \cO_\bob(t_\pi,r,r_\bob,y) = f(x,y) ] \geq 1-\epsilon.
        \]
        where $ ( a \weave b)_i
        \begin{cases}
          a_i & \textrm{if }b_i=*,\\
          b_i & \textrm{if }a_i=*,\\
          *   & \textrm{otherwise}.
        \end{cases}$
\end{definition}

We call \emph{weave} the binary operator $\weave:\set{0,1,*}^k \times \set{0,1,*}^k \rightarrow \set{0,1,*}^k$ described at the end of \cref{model:split}, that recombines the parts split among the players.

To separate this model from the one-out-of-two model, we introduce a
problem where the information about the output is naturally split
between the two players. We do so in a manner which makes computing this problem in the split
model trivial, while the fact that one of
the players must aggregate complete information about the output in the one-out-of-two model leads
to a large amount of communication.

\begin{definition}[Split identity problem]
\label{def:splitid}
    $\SplitId_n  : \{0,1\}^{n} \times \{0,1\}^{n}  \rightarrow  \{0,1\}^{n}$ is defined as  
\[
\SplitId_n(x,y)_i = \begin{cases}
  x_i & \textrm{if $i = 0 \mod 2$},\\
  y_i & \textrm{otherwise.}
\end{cases}
\]
\end{definition}

\begin{figure}[ht]
  \center
  \includegraphics[page=10]{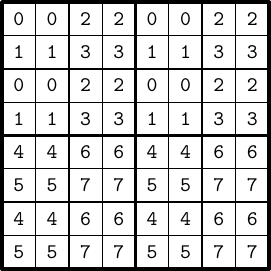}
  \label{fig:splitid}
  \caption{The communication matrix of $\SplitId_3$}
\end{figure}

\begin{theorem}
  \label{thm:splitid-oot-spt-sep}
$\forall f : \cX \times \cY \rightarrow \cZ$ with $k=  \ceil*{ \log \card{\cZ} } $ and $\epsilon > 0$ 
\begin{gather*}
R_\epsilon^\spt(f) \leq  R_\epsilon^\oot(f) \leq  R_\epsilon^\spt(f) + \floor*{ {k} / {2} } +1 , \quad \text{and}
\\
    D^\spt(\SplitId_n)  \in O(1),  \quad  \quad R^\oot_\epsilon(\SplitId_n)  \in \Omega(n).
\end{gather*}
\end{theorem}

\begin{proof}[Proof of \cref{thm:splitid-oot-spt-sep}]
  There is a small subtlety here, that the players may make the error
  of having too many or too few $*$ symbols at the end of the split
  protocol. Our proof that $R_\epsilon^\oot(f) \leq R_\epsilon^\spt(f)
  + \floor*{ {k} / {2} } +1$ must not rely on this
  assumption: we can not, for instance, say ``the player with fewer $*$
  symbols speaks first'', as this could result in an ambiguous
  protocol. 
  
  \begin{description}
    \item[Proof of $R_\epsilon^\oot(f) \leq R_\epsilon^\spt(f) + \floor*{ {k} / {2} } +1$:] Let $\Pi$ be an optimal
      split protocol. At the end of $\Pi$, Alice counts how many $*$
      symbols she would output in the split protocol. She sends a $1$
      bit if that number is greater than $\floor*{ {k} / {2}
      }$, $0$ otherwise. If she sent a $0$, she then sends
      $\floor*{ {k} / {2} }$ bits, the first of which are,
      in order, the non-$*$ symbols she would have output, in order,
      in the split protocol. If she sent a $1$, it is Bob that sends
      the first $\floor*{ {k} / {2} }$ non-$*$ bits
      that he would have output in the split protocol. In both cases,
      if there are not enough bits to send, the players append $0$'s as needed to reach $\floor*{ {k} / {2} }$ bits.

      If it is Alice that is sending the non-$*$ symbols of her split
      output, then Bob will replace the $*$ symbols in his split
      output by the bits sent by Alice before outputting it as final
      step of the one-out-of-two protocol. The situation is symmetric
      if Bob is sending his non-$*$ bits. If there are too many or not
      enough bits to replace the $*$ symbols, the bits are discarded or we just
      put $0$.

      This protocol is unambiguous (it does not rely on Alice and Bob
      not having exactly $k$ stars together) and is correct in the
      one-out-of-two model whenever the original protocol was correct
      in the split model.
    \end{description}

  The separation result again bounds the size of rectangles that do
  not make too many errors.

  \begin{description}
    \item[Split model upper bound:] Alice replaces odd positions in
      $x$ by $*$, Bob replaces even positions of $y$ by $*$. They then
      each output their resulting string, which computes
      $\SplitId_n(x,y)$ in the split model. This requires no
      communication.
    \item[One-out-of-two model lower bound:]  Consider
      $D^{\oot}_{1/4}(\SplitId_n,\mu)$, where~$\mu$ is the uniform distribution
      over $(x,y)$ such that $x_i = 0$ for odd $i$ and $y_i = 0$ for
      even $i$, and consider the communication matrix
      $\widetilde{M}_{\SplitId_n}$ of this reduced (but still total)
      problem. This reduces the number of inputs to $2^n$. Let $\Pi$
      be an optimal deterministic one-out-of-two protocol of
      communication $C=D^{\oot}_{1/4}(\SplitId_n,\mu)$.

      $\Pi$ partitions the communication matrix
      $\widetilde{M}_{\SplitId_n}$ with striped rectangles: in any
      given rectangle, the output of the one-out-of-two protocol can
      depend on either the row or on the column, but not both. But for
      our problem, every cell of the communication matrix has a
      different output, so any rectangle of width and height both at
      least $2$ makes an error in at least half its cells.

      A rectangle of width or height at most $1$ contains at most
      $2^{n/2}$ elements, therefore at most $2^{C+n/2}$ elements are
      covered by a rectangle that makes less than half error on its
      elements. Therefore at least $2^n - 2^{C+n/2}$ inputs are
      covered by rectangles with at least $1/2$ error, so $\Pi$ makes
      error at least $2^{-n}\cdot \frac 1 2 \parens*{ 2^n - 2^{C+n/2}
      }$. This error has to be less than $\frac 1 4$, so:

      \[ \frac 1 4 \geq 2^{-n}\cdot \frac 1 2 \parens*{ 2^n - 2^{C+n/2}
      } \Rightarrow C \geq n/2 -1 \]

      Which completes our proof that $R^{\oot}_{1/4}(\SplitId_n) \geq D^{\oot}_{1/4}(\SplitId_n,\mu)
      \in \Omega(n)$.
  \end{description}
\end{proof}

\subsubsection{The XOR model}

In our final model, the players both output a $k$ bit string at the
end of the protocol. A computation correctly computes the value of $f(x,y)$
 when the bit-wise XOR of the two strings is equal to $f(x,y)$.


\xormodel*

The XOR model is separated from the one-out-of-two model by the
following function:

\begin{definition}
\label{def:xor}
$\XOR_n : \{0,1\}^n \times \{0,1\}^n  \rightarrow  \{0,1\}^n $ is defined by 
 $\XOR_n(x,y) =  (x_i \oplus y_i)_{i \in [n]}$.
\end{definition}

\begin{figure}[ht]
  \center
  \includegraphics[page=11]{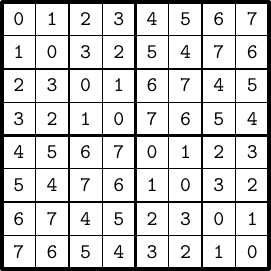}
  \label{fig:xor}
  \caption{The communication matrix of $\XOR_3$}
\end{figure}

\begin{theorem}
\label{sep:oot-xor}
  $\forall f : \cX \times \cY \rightarrow \cZ$ with $k=  \ceil*{ \log \card{\cZ} } $  and $\epsilon > 0$ ,
\begin{gather*}
  R_\epsilon^\xor(f) \leq R_\epsilon^\spt(f) \leq  R_\epsilon^\oot(f) \leq R_\epsilon^\uni(f) \leq  R_\epsilon^\xor(f) + k, \quad \text{and} \\
    D^\xor(\XOR_n) =0, \quad \quad
    R^\spt_\epsilon(\XOR_n)  \in \Omega(n).
    \end{gather*}
\end{theorem}

\begin{proof}[Proof of \cref{sep:oot-xor}]
  \begin{description}
    \item[XOR model upper bound:] Alice and Bob can just each output
      their input, which requires no communication.
    \item[Split model lower bound:] Let us consider
      $D^{\spt}_{1/4}(\XOR_n,\mu)$ where $\mu$ is the uniform
      distribution. Let $\Pi$ be an optimal deterministic
      one-out-of-two protocol of communication
      $C=D^{\spt}_{1/4}(\XOR_n,\mu)$.

      $\Pi$ partitions the communication matrix $M_{\XOR_n}$ into
      $2^C$ rectangles. Let us first assume that in each rectangle,
      each bit of the output is output by a fixed player. We will see
      later that our argument still holds without this assumption.

      In each of the $2^C$ rectangles, one of the players has to
      output less than $n/2$ bits of the output. Let us consider a
      rectangle where Bob outputs at most half the bits of the
      output. Then, on a given row of this rectangle, there can be at
      most $2^{n/2}$ different outputs. But the $\XOR_n$ problem is
      such that on a given row, all cells have a different output. We
      will argue that this bounds the size of the rectangles that do
      not make a lot of error.

      Let a rectangle contain at least $2^{3n/2+1}$ elements. Since a row or column contains at most $2^n$ elements, such a rectangle contains at least $2^{n/2+1}$ rows and columns. Therefore, the player
      that outputs at most half the bits of the output in the split
      model will output at most $2^{n/2}$ different strings on a given
      row or column that contains more than $2^{n/2+1}$ different
      values, so the rectangle has error on at least half of its
      elements.

      If the players do not always split the outputs bits in the same
      way, consider the largest set of rows such that Alice outputs a
      given subset of the output bits, and the largest set of columns
      such that Bob outputs a given subset of the output bits. If the
      sets of output bits that Alice and Bob output on those rows and
      columns are not the complement of each other, the rectangle is
      in error on at least half of its elements. If the sets correctly
      partition the output bits, we do the same argument as before:
      let us assume that Bob outputs at most half the bits in the
      subrectangle we defined. Then no more than $2^n$ cells can be
      correct in any row of this subrectangle, and rows outside of the
      subrectangle are also mostly error, therefore the rectangle has
      error on at least half of its elements.

      At most $2^{C+3n/2+1}$ elements are in rectangles with error
      strictly less than half, so the error made by the protocol is at
      least $\frac 1 2 \cdot 2^{-2n} \parens*{ 2^{2n} - 2^{C+3n/2+1} }$. The error has to be less than $\frac 1 4$, so:
      \[C \geq n/2 -2\]
      Which completes our proof that $R^{\spt}_{1/4}(\XOR_n) \geq
      D^{\spt}_{1/4}(\XOR_n,\mu) \in \Omega(n)$.
  \end{description}
  
\end{proof}

\subsection{Relations between models}

The next proposition summarizes the relations between models seen in
\cref{thm:eqout-open-loc-sep,thm:id-loc-uni-sep,thm:condid-uni-oot-sep,thm:splitid-oot-spt-sep,sep:oot-xor}.

\begin{proposition}
  \label{prop:hierarchy}
  $\forall f : \cX \times \cY \rightarrow \cZ$ with $k=  \ceil*{ \log \card{\cZ} } $  and $\epsilon > 0$ we have:
  \begin{align}
    \label{subprop:hierarchy}
    R_\epsilon^\op(f) \geq R_\epsilon^\loc(f) 
            & \geq \max\parens*{R_\epsilon^\ali(f), R_\epsilon^\bob(f) } \\
            & \geq \min\parens*{R_\epsilon^\ali(f), R_\epsilon^\bob(f) }
              = R_\epsilon^\uni(f) \notag \\
            & \geq R_\epsilon^\oot(f) \geq R_\epsilon^\spt(f) \geq R_\epsilon^\xor(f) \notag \\
    \label[subproposition]{subprop:loc-leq-ali-bob}
    R_{2\epsilon}^\loc(f) & \leq R_\epsilon^\ali(f) + R_\epsilon^\bob(f) \\
    \label{subprop:op-leq-uni-log-Z}
    R_\epsilon^\op(f) & \leq R_\epsilon^\uni(f) + k\\
    \label{subprop:op-leq-oot-log-Z}
    R_\epsilon^\op(f) & \leq R_\epsilon^\oot(f) + k +1\\
    \label{subprop:oot-leq-spt-log-Z}
    R_\epsilon^\oot(f) & \leq R_\epsilon^\spt(f) + \lceil k / 2 \rceil + 1.\\
    \label{subprop:uni-leq-xor-log-Z}
    R_\epsilon^\uni(f) & \leq R_\epsilon^\xor(f) + k.
  \end{align}
  The same statements hold for deterministic communication and
  communication with private randomness only. All statements except
  \cref{subprop:loc-leq-ali-bob} also hold for
  relations and nondeterministic communication.
\end{proposition}

\Cref{prop:hierarchy} shows that the models form a natural
hierarchy and can be ordered from most to least communication
intensive.
We also summarize this hierarchy in \cref{fig:hierarchy}, in the main text.
This figure also displays separating problems other than those in this section, in Appendix.

\section{Summary of our results}

In this section, we summarize the results in this paper. 
\Cref{tab:all-problems} summarizes the problems we have studied which show gaps between the different 
output models.
\Cref{tab:cc-hxor} summarizes the bounds on $\GapMajX$ in various models.
\Cref{tab:all-schemes} summarizes error reduction bounds and derandomization.

\begin{table}[htbp]
\centering
\newcommand{\vpad}{\rule[-0.4em]{0pt}{1.4em}}
\begin{tabularx}{\linewidth}{|c|Z{1.425}|Z{0.65}|Z{0.75}|Z{1.05}|Z{1.125}|}
\hline
& open & local & unilateral & 1-out-of-2 & XOR \\ \hline
$\EQout_n$ & $R^\cM_{\scriptscriptstyle 1/3} \in \Theta(n)$ &\multicolumn{4}{c|}{$R^\cM_{\scriptscriptstyle 1/3} \in \Theta(1)$} \\ \hline
\vpad
$t-\INT_n$ & $R^\cM_{\scriptscriptstyle 1/3} \in \Theta(t\cdot \log (n))$ &\multicolumn{4}{c|}{$R^\cM_{\scriptscriptstyle 1/3} \in \Theta(t)$} \\ 
\hline
\vpad
$\id^\ali_n$ & \multicolumn{2}{c|}{$R^\cM_{\scriptscriptstyle 1/3} \in \Theta(n)$} &\multicolumn{3}{c|}{$D^\cM = 0$} \\ \hline
\vpad
$\CondId_n$ & \multicolumn{3}{c|}{$R^\cM_{\scriptscriptstyle 1/3} \in \Theta(n)$} &\multicolumn{2}{c|}{$D^\cM = 2$} \\ \hline
\vpad
$\MAX_n$ & \multicolumn{3}{c|}{$R^\cM_{\scriptscriptstyle 1/3} \in \Theta(n)$} &\multicolumn{2}{c|}{$R^\cM_{\scriptscriptstyle 1/3} \in \Theta(\log(n)$} \\ \hline
\vpad
$t-\FtFD_n$ & \multicolumn{3}{c|}{$R^\cM_{\scriptscriptstyle 1/3} \in \Theta(\log(n))$} &\multicolumn{2}{c|}{$R^\cM_{\scriptscriptstyle 1/3} \in \Theta(\log(t)+\log\log(n))$} \\ \hline
\vpad
$\XOR_n$ & \multicolumn{4}{c|}{$R^\cM_{\scriptscriptstyle 1/3} \in \Theta(n)$} & $D^\cM = 0$\\ \hline
\vpad
$\GapMajX[N,k,{\scriptscriptstyle 1/3}]$ & \multicolumn{4}{c|}{$R^\cM_{\scriptscriptstyle 1/3} \in \Theta(k)$} & $R^\cM_{\scriptscriptstyle 1/3} = 0$ \\ \hline
\vpad
$\GapMajX[N,k,{\scriptscriptstyle 2/5}]$ & \multicolumn{4}{c|}{$R^\cM_{\scriptscriptstyle 1/3} \in \Theta(k)$} & $R^\cM_{\scriptscriptstyle 1/3} \in O(1)$ \\ \hline
\end{tabularx}
\caption{Summary of the communication complexities of our separating
  problems in all models. The definitions of the problems and the proofs are in \cref{app:outputs-models,app:separating-problems}. 
In this table, $n$ is the input length, $k$ is the output length, $ \mdl$ is an output model, 
 $\mdl \in \set{\op, \loc, \ali, \bob, \uni, \oot, \xor }$, and~$t$ is the Hamming weight of an instance.
}
\label{tab:all-problems}
\end{table}

\begin{table}[tbp]
\centering
\newcommand{\vpad}{\rule[-0.4em]{0pt}{1.4em}}
\newcommand{\vvpad}{\rule[-0.75em]{0pt}{2.0em}}

\begin{tabularx}{0.8\linewidth}{|c|c|Z{1.0}|}
  \hline
  \multicolumn{2}{|c}{} & \multicolumn{1}{c|}{ \textbf{Upper bounds} } \\
  \hline
  \multirow{4}{*}{$\epsilon' \geq \epsilon$} 
        & \vpad $R^\xor_{\epsilon'}$ & $0$ \\ \cline{2-3}
        & \vpad $R^{\xor,\priv}_{\epsilon'}$ & $ \log(N)$ \\ \cline{2-3}
        & \vpad $R^{\op}_{\epsilon'}$ &     $ 2k$ \\ \cline{2-3}
        & \vpad $R^{\op,\priv}_{\epsilon'}$ & $ 2k+\log(N)$ \\ \hline
  \multirow{1}{*}{$0 < \epsilon' < \epsilon$}
        & \vvpad $R^\xor_{\epsilon'}$ 
            & $ O\parens*{
            \min\parens*{
            C_{\epsilon,\epsilon'}, N + \log\parens*{ \frac 1 {\epsilon'} }}
            }$\\ \hline 
  \multirow{1}{*}{$ \epsilon'=0$}            
        & \vpad $D^\uni$ & $ (2\epsilon N+1)k$\\ \hline
  
\end{tabularx}
\caption{Upper bounds on $\GapMajX$, proofs in \cref{app:cc-hxor}. In this table, $N,k,\epsilon$ are the parameters
of the Gap Majority problem,  and
$\epsilon'$ is the error parameter.
}
\label{tab:cc-hxor}
\end{table}

\begin{table}[htbp]
\centering
\newcommand{\vpad}{\rule[-0.45em]{0pt}{1.5em}}
\newcommand{\vvpad}{\rule[-0.85em]{0pt}{2.15em}}
\begin{tabularx}{\linewidth}{|c|Yr|}
\hline
\multicolumn{3}{|c|}{\textbf{Error reduction}} \\ \hline
model & Upper bounds &  
\hspace{-3cm}(condition)
\\ \hline
open &\multirow{3}{*}{$R_{\epsilon'}(f) \leq C_{\epsilon,\epsilon'} \cdot R_\epsilon(f)$} &\\
local & & \\ 
unilateral & & \\ \hline
1-out-of-2 
& \vpad {$R_{\epsilon'}(f) \leq C_{\epsilon,\epsilon'} (R_\epsilon(f) + 1) + C'_{\epsilon,\epsilon'}$}& \\ \hline
split
& \vpad {$ R_{\epsilon'}(f) \leq C_{\epsilon, \epsilon'} 
R_\epsilon(f) 
   + O \parens*{ C_{\epsilon, \epsilon'} } $}& \\ \hline
\multirow{2}{*}{XOR}
& \vpad {$ R_{\epsilon'}(f) \leq C_{\epsilon, \epsilon'} 
R_\epsilon(f) 
   + O \parens*{ C_{\epsilon, \epsilon'} } $}& \\ \cline{2-3}
& \vpad {$R_{\epsilon'}(f) \leq 50 \ln \parens*{\frac {12} {\epsilon'} }
R_\epsilon(f) +C_{\epsilon, \epsilon'} 
R_\epsilon(g) + O \parens*{ C_{\epsilon, \epsilon'}+ \log (k) }$}
 & {$(f=g^{\otimes k})$}\\ \hline
\multicolumn{3}{|c|}{\textbf{Derandomization}} \\ \hline
model & Upper bounds & 
\hspace{-3cm}(condition)
\\ \hline
open &
\multirow{2}{*}{
{$D(f) \in O\parens*{ 2^R \parens*{R + \log \parens[\big]{ \frac 1 {\frac 1 2 - \epsilon}}}}$}} & \\
local & &\\ \hline
unilateral 
& \vvpad {$D(f) \in O\parens*{ 2^R \parens*{R + \log \parens[\big]{ \frac 1 {\frac 1 2 - \epsilon}}}}$} &\\ \hline
\multirow{2}{*}{1-out-of-2} 
& \vvpad {$D(f) \in O\parens*{ 2^R \parens*{R + \log \parens[\big]{ \frac 1 {\frac 1 2 - \epsilon}}}}$} &\\ \cline{2-3}
& \vvpad {$D(f) \in O\parens*{ 2^R \parens*{R + \log \parens[\big]{ \frac 1 {\frac 1 2 - \epsilon}}} + \log(k)}$} &\\ \hline
\multirow{2}{*}{split} 
& \vvpad {$D(f) \in O\parens*{ 2^R \parens*{R + \log \parens[\big]{ \frac 1 {\frac 1 2 - \epsilon}}} + k}$} &\\ \cline{2-3}
& \vvpad {$D(f) \in O\parens*{ 2^R \parens*{R + \log \parens[\big]{ \frac 1 {\frac 1 2 - \epsilon}}} + 2^R \parens*{\frac 1 2 - \epsilon }^{-2} k}$} &\\ \hline
XOR &
\vvpad {$D(f) \in O\parens*{ 2^R \parens*{R + \log \parens[\big]{ \frac 1 {\frac 1 2 - \epsilon}}} + 2^R \parens*{\frac 1 2 - \epsilon }^{-2} k}$} &\\ \hline
\end{tabularx}
\caption{Summary of our error reduction and derandomization schemes. 
In all statements above, $f$~is a function  whose output length is $k$,
$\epsilon$ is the starting error parameter, $\epsilon'$ is the target error parameter,
$R=R_\epsilon^\mdl(f)$,
 $C_{\epsilon,\epsilon'} \in O\parens*{
\epsilon {
\parens*{ \frac 1 2 - \epsilon } ^{-2}
}{\log \parens*{\frac 1 {\epsilon'}}}
}$ and
$C'_{\epsilon,\epsilon'} \in O\parens*{\log\parens*{ \frac 1
{\epsilon'} } + \log \parens*{ \frac 1 { \frac 1 2 - \epsilon } }}$.
}
\label{tab:all-schemes}
\end{table}

The upper bounds on the \emph{Gap Majority} problem,
are summarized in \cref{tab:cc-hxor}.
 We conjecture a matching lower bound to our
stated deterministic $O(\epsilon N k)$ upper bound. Studying the
communication complexity of this problem is of theoretical interest,
as we have seen in this paper that fundamental results in
communication complexity, namely error reduction and derandomization,
are related to the $\GapMajX$ problem in the XOR model. Improving
the deterministic upper bound on $\GapMajX$ would yield a better
derandomization result through \cref{thm:derand-xor}. 
Similarly, improving the randomized upper bounds could improve
error reduction through
\cref{lem:amp-xor-hxor}. Conversely, considering that we have an
upper bound of $\log(N)$ on the private coin XOR communication
complexity of $\GapMajX$, proving a $\Omega(Nk)$ lower bound on its
deterministic communication complexity would indicate that our
derandomization theorem in the XOR model
(\cref{thm:derand-xor}) is close to tight.



\section{The weak partition bound}
\label{app:wprt}

The weak partition bound can be used to obtain lower bounds on the open model. We use it to show that  this model is very sensitive to the number of ``non-trivial'' or ``typical'' outputs,  
 those that occur frequently enough, in a sense that is made precise in \cref{def:xi}.

\begin{definition}[Weak partition bound~\cite{FontesJKLLR16}]
\label{rect-def:wprt-nonbool-total}
We define (using the notation $\beta = \sum_{x,y} \beta_{x,y}$)
\begin{align}
\wprt^\mu_\epsilon(f) = &\max_{\alpha \geq 0, \, \beta_{xy} \geq 0}  && (1-\epsilon)\alpha - \beta 
	\notag
	\\
         &\text{subject to :} 
        &&  \alpha\mu(R\cap f^{-1}(z)) - \beta(R) \leq 1 &\quad \forall R,z,
        \label{rect-wprt-constr-1}
        \\
&&&  \alpha\mu_{xy}-\beta_{xy} \geq 0 & \forall (x,y).
	\label{rect-wprt-constr-pos}
\end{align}
The non-distributional weak partition bound of $f$ is
$\wprt_\epsilon(f) = \max_\mu \wprt^\mu_\epsilon(f)$.
\end{definition}

Note that the definition we have here is slightly
different from the one given by Fontes et al~\cite{FontesJKLLR16}. The
two formulations are equivalent for Boolean functions, which was the
setting considered in that paper.

\begin{proposition}[\cite{JainK10,FontesJKLLR16}]
Let $0 < \epsilon < 1/2$ and let $f: \cX \times \cY \rightarrow \cZ$ be a function. Then,
\[ \log\parens*{\wprt_\epsilon(f)}
\leq \log\parens*{\prt_\epsilon(f)}
\leq R_\epsilon^\op(f).\]

The right-hand side is from~\cite{JainK10} and the
left-hand side from~\cite{FontesJKLLR16}.
\end{proposition}

We then introduce the notion of $\epsilon$-Minimum set of outputs with respect to a distribution $\mu$. 
Let us abuse notation and write $\mu(z)$ for $\mu(f^{-1}(z))$ when  there is no need to specify which $f$ we are implicitly referring to.
  
\begin{definition}
\label{def:xi}
Let $\cZ$ be the set of outputs of a function $f: \cX \times \cY \rightarrow \cZ$.

  Let us further consider that $\cZ = \set{z_1, z_2, \ldots, z_n}$ is
  sorted with respect to $\mu$, that is :
  \[i \leq j \Rightarrow \mu(z_i) \geq \mu(z_j).\]

  Then $\xi^\mu_\epsilon(f)$ is defined as:
  \[ \xi^\mu_\epsilon(f) = \min \set*{ k \;\Big|\, \sum_{i=1}^{k}\mu(z_i) \geq 1-\epsilon  }.\]
  
\end{definition}

\begin{theorem}
\label{thm:wprt-xi}
Let $0 < \epsilon < 1/2$, let $f: \cX \times \cY \rightarrow \cZ$ be a function and  let $\mu$ be a distribution over $\cX \times \cY$. Then,
\[ \xi^\mu_\epsilon(f) - 1 \leq \wprt^\mu_\epsilon(f). \]
\end{theorem}

\begin{proof}[Proof of \cref{thm:wprt-xi}]
  Sort the set of outputs with respect to $\mu$ (i.e., $z_1 \leq z_2 \leq
  \ldots \leq z_n$) and set $z_{\min} =
  z_{\xi^\mu_\epsilon(f)}$. Consider the following assignment of
  variables :
\[
    \alpha = \frac 1 {\mu(z_{\min})}, \qquad \qquad
        \beta_{xy} = \max \parens*{ 0 , \mu_{xy} \cdot \parens*{ \alpha - \frac 1 {\mu( f(x,y))} } }.
  \]

  Then the first constraint of $\wprt^\mu_\epsilon$ is satisfied. Indeed, let  $z$ be s.t.\ $\mu(z) \leq \mu(z_{\min})$ (and so $\beta_{xy}=0$ for all $(x,y) \in f^{-1}(z)$). Then for all for all $R$:
  \begin{align*}
    & \alpha \cdot \mu(R\cap f^{-1}(z))-\beta(R) & &\\
     \leq\  &  \alpha \cdot \mu(R\cap f^{-1}(z))-\beta(R \cap f^{-1}(z)) & &\\
     =\ & \alpha \cdot \mu(R\cap f^{-1}(z)) \leq \alpha \mu(z) =  \frac {\mu(z)} {\mu(z_{\min})} \leq 1.
    \intertext{When $z$ is s.t.\ $\mu(z) > \mu(z_{\min})$, for all $R$:}
        & \alpha \cdot \mu(R\cap f^{-1}(z))-\beta(R) & &\\
     \leq\  &  \alpha \cdot \mu(R\cap f^{-1}(z))-\beta(R \cap f^{-1}(z)) & &\\
     =\ & \alpha \cdot \mu(R \cap f^{-1}(z)) - 
     \parens*{\alpha - \frac{1}{\mu(z)}} \mu(R \cap f^{-1}(z))
     \\ =\ & \frac{\mu(R \cap f^{-1}(z))}{\mu(z)} \leq 1.
  \end{align*}

  The second constraint is satisfied as well:
%
\[    \forall x,y : \quad \alpha \mu_{xy} - \beta_{xy} 
      = \alpha \mu_{xy} - \max \parens*{ 0 , \mu_{xy} \parens*{ \alpha - \frac 1 {\mu( z^{xy})} }} \geq 0.
      \]

  And the value of this feasible solution is:
  \begin{flalign*}
    (1-\epsilon) \alpha - \beta & = (1-\epsilon) \frac {1}{\mu(z_{\min})} - \sum_{z:z=z_i,\,i < \xi^\mu_\epsilon(f)} \beta(z_i)\\
    & = \parens*{1-\epsilon - \sum_{z:z=z_i,\,i < \xi^\mu_\epsilon} \mu(z_i) } \frac {1}{\mu(z_{\min})} + \xi^\mu_\epsilon(f) - 1 \\
    & \geq  \xi^\mu_\epsilon(f) - 1.
  \end{flalign*}
\end{proof}


\section{Error reduction}
\label{app:proofs-error-reduction}

\subsection{Proof of the random graph lemma}

The proof of the random graph lemma stated in
\cref{sec:hxor-amp} and used to solve $\GapMajX$ is a simple variation of a
result of Erd\H{o}s and Rényi~\cite{ErdosR60}. The result they proved
is in a model of random graphs where a fixed number of edges are
picked randomly from the set of all possible edges, while we are
interested in a model of random graphs where each edge is picked with
a fixed probability $p$ independently of other edges. The two models
are known to have essentially similar asymptotic behaviours. Readers
interested in the theory of random graphs might refer 
to~\cite{Bollobas01}.

\begin{proof}[Proof of \cref{lem:connected-component}]
We observe as in~\cite{ErdosR60} that if no connected component of
more than $(1-\alpha)n$ vertices exists, then we can partition the
vertices into two disconnected sets of size $n_0$ and $n_1$ such that
$\frac \alpha 2 n \leq n_0 \leq n_1 \leq \parens*{1- \frac \alpha
2} n$.

Given a partition of the vertices into sets of size $n_0$ and $n_1$,
the probability that those two sets are disconnected is
$(1-p(n))^{n_0n_1}$. With $p(n) = \frac c n$, and since
there are less than $2^n$ possible partitions, the probability that
there is no connected component of more than $(1-\alpha)n$
vertices is bounded by:
\begin{align*}
  2^n\parens*{1-\frac c n}^{n_0n_1} &\leq 2^n e^{-c\frac{n_0 n_1} n} \leq 2^n e^{-c\frac \alpha 2 \parens*{1-\frac \alpha 2 }n} = e^{\parens*{\ln(2) - \frac \alpha 2 \parens*{1-\frac \alpha 2} c}n}
\end{align*}
\end{proof}

\subsection{Error reduction up to the XOR model}
\label{app:error-reduction-up-to-XOR}

\subsubsection{Error reduction in the one-out-of-two model}
The one-out-of-two model is already non-trivial. 
If we repeat the protocol, in some runs Alice will output, and in others, Bob will output,
and it is possible that on both sides, the majority output is incorrect. 
A trivial way to reduce error in this model would be to convert
the one-out-of-two protocol to the unilateral model 
(\cref{prop:hierarchy}) and
apply \cref{thm:amp-kn97}, to
obtain 
$R_{\epsilon'}^\oot(f) \leq C_{\epsilon,\epsilon'} \cdot ( R^\oot_\epsilon(f)
           + k 
).
$

We prove that the additional dependency on the output length $k$ can be 
 removed.
 We show that the players
can narrow down the number of candidates for the majority outcome to at most four. 
Hashing is used to single out the winning outcome with high probability.
\begin{theorem}
  \label{thm:amp-oot}
  Let $0<\epsilon'<\epsilon<\frac 1 2$, $C_{\epsilon,\epsilon'} =  \frac {2 \epsilon (1 - \epsilon)}{\parens*{ \frac 1 2 - \epsilon}^2} \ln \parens*{\frac 4 {\epsilon'} }$
  and
  $C'_{\epsilon,\epsilon'} \leq 18 + 4 \log \parens*{ \frac 1
  {\epsilon'} } + 4 \log \parens*{ C_{\epsilon,\epsilon'} }$. 
 For all functions $f :
  \cX \times \cY \rightarrow \cZ$, 
  $$R_{\epsilon'}^\oot(f) \leq C_{\epsilon,\epsilon'} (R^\oot_\epsilon(f) + 1)
  + C'_{\epsilon,\epsilon'}.$$
\end{theorem}

\begin{proof}[Proof of \cref{thm:amp-oot}]
Fix a one-out-of-two protocol for $f$ with error at most $\epsilon$
and apply \cref{prop:oot-one-speaker} so that we now
have a one-out-of-two communication protocol and mappings such that
exactly one player speaks at the end in any execution.
Using Hoeffding's inequality (\cref{lem:hoeffding}), if the
players make $T=\ceil{ 8 \epsilon (1 - \epsilon)\parens*{ \frac 1 2 -
  \epsilon}^{-2} \ln \parens*{ \frac 4 {\epsilon'} } }$ executions, then with
probability at least $1-\frac {\epsilon'} 2$, in at least $\frac 1 2 + \frac 1 2 \parens*{\frac 1 2 - \epsilon}>\frac 1 2$ of the
player's executions, one of them outputs $f(x,y)$ (and the other
remains silent).

The players want to identify the correct output. We argue that they can do it with very little extra communication and error. 
Observe that if a value is output in strictly more than half of the above executions, it must have been output strictly more than $T/4$ 
times by one of the two players.
Thus each player only needs to consider outputs that appear
stricly more than $T/4$ times on its side.
Let us call $(z^\ali_i)_{i
  \in [n_\ali]}$ and $(z^\bob_j)_{j \in [n_\bob]}$ the outputs identified as
candidates for $f(x,y)$ respectively on Alice's and Bob's side, $n_\ali$ and $n_\bob$ being
the number of candidates on each side.  
Since there are $T$ executions, each candidate is output strictly more than $T/4$ times and one value is output stricly more than $T/2$ times, there are at most 3 candidates and 
$n_\ali \leq 2$ and $n_\bob \leq 2$. 

The players use their public randomness to pick a random hash function $h : \cZ \rightarrow [m]$ where $m$ is to be chosen so that, with high probability, there are
no collisions among the candidates $(z^\ali_i)_{i
  \in [n_\ali]}$ and $(z^\bob_j)_{j \in [n_\bob]}$ selected by the players. Since the probability of a given
collision is $\frac 1 m$ and there are  $n_\ali+n_\bob \leq \binom{4}{2}$ pairs of
candidates, taking $m =  \ceil*{ \frac {12} {\epsilon'} } \geq \frac {12} {\epsilon'}$ guarantees that
such a collision only occurs with probability~$\leq \frac {\epsilon'}
2$.%
The players then exchange the hashes $h_1, \ldots, h_{n_\ali+n_\bob}$ of
their candidates (corresponding to $h(z^\ali_i)$ and $h(z^\bob_j)$ for $i
\in [n_\ali]$ and $j \in [n_\bob]$) with $4 \ceil{ \log(m) }$ bits of
communication. 
For each $k$, Alice computes $\alpha_k = \card{\set{i:h(z^\ali_i) = h_k}}$ and Bob computes $\card{\set{j:h(z^\bob_j) = h_k}}$.
 Alice sends her
counts $(\alpha_k)_{k=1,\dots,n_\ali+n_\bob}$ to Bob (with communication $\leq 4\ceil*{ \log(T)}$). Bob
replies with $k \in [n_\ali+n_\bob]$ such that $h_k$ is the hash that most
outputs hash to through $h$. If that hash is the hash of a candidate
$z^\ali_i$ of Alice, she outputs this candidate, otherwise it is Bob who
outputs his corresponding candidate. Adding the errors due to 
deviation (Hoeffding) and to collisions, this protocol makes at most $\epsilon'$ error.
\end{proof}

\subsubsection{Error reduction in the split model}

 Remarkably, error reduction in the split model can be achieved
   very similarly to the scheme for the XOR model. 
   Notice that it is not sufficient to apply the XOR scheme by replacing stars with zeros, since 
   the {\em output} should be split as well (whereas the output in the XOR scheme is not necessarily of this form). More precisely, applying \cref{thm:error-reduction-no-k} would show $R^{\xor}_{\epsilon'}(f) \leq 
   C_{\epsilon, \epsilon'} \cdot R^\spt_\epsilon(f) 
   + O \parens*{ C_{\epsilon, \epsilon'} }$.
   
   The key observation we used to
  reduce error in the XOR model was that when two rows~$i$ and~$j$
  of the $\GapMajX$ matrix XORed to the same string, i.e., $X_i \oplus Y_i
  = X_j \oplus Y_j$, we observed that $X_i \oplus X_j = Y_i \oplus
  Y_j$. This allowed us to test whether two rows XORed to the same
  string by making one equality test on two locally-computable strings.
  We call this local operation a ``compatibility gadget'',
  that is,  a function $g$ that the players apply
  locally to pairs of rows, such that the problem of testing $X_i \oplus Y_i = X_j \oplus Y_j$
   reduces to testing equality between $g(X_i,X_j)$ and
  $g(Y_i,Y_j)$. In the XOR model, the compatibility gaget  $g$ was  just a bit-wise XOR. The bitwise compatibiklity gadget is illistrated in 
  \cref{fig:compatibility-xor}.

\begin{figure}[ht]
  \center
  \begin{subfigure}[t]{0.2\textwidth}
   \center
   \includegraphics[page=1,width=\textwidth]{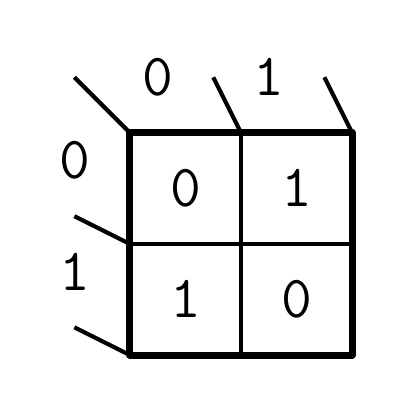}
   \caption{XOR  gadget}
   \label{fig:compatibility-xor}
  \end{subfigure}
  ~
  \begin{subfigure}[t]{0.3\textwidth}
   \center
   \includegraphics[page=2,width=0.8\textwidth]{compatibility}
   \caption{Alice's split  gadget $g_\ali$}
   \label{fig:compatibility-split-A}
  \end{subfigure}
  ~
  \begin{subfigure}[t]{0.3\textwidth}
   \center
   \includegraphics[page=3,width=0.8\textwidth]{compatibility}
   \caption{Bob's split  gadget $g_\bob$}
   \label{fig:compatibility-split-B}
  \end{subfigure}  
  \caption{The matrices of the compatibility gadgets for the XOR and the split models.}
    \label{fig:compatibility}
\end{figure}

It turns out that we can do the something similar in the split model, with a
slight change. Instead of both players applying the same gadget on pairs of rows
 before testing for equality, they each apply a different
gadget. The functions they apply bit-wise to pairs of rows are the
transformations $g_\ali$ and $g_\bob$ represented in
\cref{fig:compatibility-split-A,fig:compatibility-split-B}. 
The functions are chosen so that the following property holds.

\begin{proposition}
  \label{prop:split-compatibility}
  For all $X_i$, $X_j$, $Y_i$, and $Y_j \in \set{0,1,*}^k$, and $g_\ali, g_\bob$ described in \cref{fig:compatibility-split-A,fig:compatibility-split-B},
  \[X_i \weave Y_i =
  X_j \weave Y_j \Leftrightarrow g_\ali(X_i,X_j) = g_\bob(Y_i,Y_j)\]
\end{proposition}

These functions  capture when
a pair of rows output the same result: if Alice outputs
two stars in some position of $X_i$ and $X_j$, then Bob needs to be
outputting two $0$s or two $1$s in the same position in his strings
($Y_i$ and $Y_j$). Similarly, if at some index Alice outputs a star in
row $X_i$ and a $0$ in row $X_j$, then at this same index, Bob needs to
output a $0$ in $Y_i$ and a star in $Y_j$ so that the two rows
yield the same result.

\Cref{prop:split-compatibility} implies that error-reduction in the split model reduces to solving $\GapMAJ$ combined with the weave gadget ($\weave$), 
in the same way that error reduction in the XOR model reduced to solving $\GapMajX$.
We  obtain the following similar result
\cref{thm:error-reduction-split}.

\begin{theorem}\label{thm:error-reduction-split}
  Let $0 < \epsilon' < \epsilon < \frac 1 2$, $C_{\epsilon,\epsilon'}
  = 8\epsilon {\parens*{ \frac 1 2 - \epsilon } ^{-2}}{\ln \parens*{\frac 4 {\epsilon'}}}$.  For all $f : \cX \times \cY
  \rightarrow \zo^k$,
 \begin{align*}
  R^{\spt}_{\epsilon'}(f) \leq 
   C_{\epsilon, \epsilon'} \cdot R^\spt_\epsilon(f) 
   + O \parens*{ C_{\epsilon, \epsilon'} }. & \label{eq:amp-split-no-k} 
  \end{align*}
\end{theorem}



\section{Error reduction for direct sum problems}
\label{app:direct-sum}

This section gives the full proof of \cref{thm:error-reduction-direct-sum} which gives an error reduction scheme for functions of the form $f=g^{\otimes k}$.

\begin{proof}[Proof of \cref{thm:error-reduction-direct-sum}]
Consider an XOR protocol for $f=g^{\otimes k}$ with error at most $\epsilon$, together
with a protocol for $g$ with error at most $\epsilon$.
  The protocol to achieve error $\epsilon'$ proceeds as follows.

  \begin{description}
  \item[Step 1:] [Restrict to at most two candidates.] The players run
    the XOR protocol for $f$ for a total of
    $T_{\epsilon'}=50 \ln \parens*{\frac {12} {\epsilon'} }$
    iterations. Let $(a_i,b_i)$ be what Alice and Bob would have
    output on the $i^{th}$ iteration. As in \hyperlink{step:many-iterations-f}{Step~1} of the proof of
    \cref{thm:hxor-no-k}, with high probability, $\card{\set{i: a_i
    \oplus b_i = f(x,y)}} \geq \frac 2 5 T_{\epsilon'}$.

As in \hyperlink{step:first-equality-batch}{Step~2} of the proof of \cref{thm:hxor-no-k}, the players
then solve random $\EQ$ instances to find large subsets of iterations
with the same computed value. With high probability ($\geq 1 -
\frac{\epsilon'} 6$), they compute at most $O(T_{\epsilon'})$ instances
of $\EQ$, with $\frac {\epsilon'} 6$ error.

With high probability ($\geq 1 - 3 \cdot \frac {\epsilon'} 6$), the
players should have identified either one or two sets of at least
$\frac {11}{30}T_{\epsilon'}$ iterations such that all iterations in a
set computed the same value. If only one such large set was found, the
players output $a_i$ and $b_i$ where $i$ is the index of an arbitrary
iteration in this large set. Otherwise, let $i_1$ and $i_2$ be 
indices, each one representing one of the two large sets.

\item[Step 2:] [Find a critical index $l$.]
The players will either output as in the $i_1^{th}$ or the $i_2^{th}$
iteration. To decide between the two, they find the first difference
between $a_{i_1} \oplus a_{i_2}$ and $b_{i_1} \oplus b_{i_2}$. This
yields an index $l \in [k]$ where the two possible outputs differ. We
call this a critical index.

\item[Step 3.] [Solve GHD on the critical index $l$.]
We XOR-compute the $l^{th}$ bit of $f$ $ C_{\epsilon,\epsilon'} $
times. This gives an instance of  Gap Hamming Distance of size $C_{\epsilon,\epsilon'}$
whose solution determines
 the $l^{th}$ bit of the correct output, with high probability.
The players determine which iteration, $i_1$ or~$i_2$, was correct
on the $l^{th}$ bit, and output according to that iteration.
\end{description}
Altogether, we get the following upper bound on computing $f$ with error $\epsilon'$.
  \label{eqn:amp-xor-version}
  \begin{align*} R^{\xor}_{\epsilon'}(f) 
  & \leq
  \parens*{ 50 \ln \parens*{\frac {12} {\epsilon'} } }\cdot R^\xor_\epsilon(f) 
  + R^\loc_{\epsilon'/6}\parens*{ \EQ_k^{\otimes O(T_{\epsilon'})} }
	+ R^\loc_{\epsilon'/6} \parens*{ \FtFD_k } \\
  & \qquad\quad+ C_{\epsilon, \epsilon'}
  \cdot R^\xor_\epsilon(g) 
   + R^\loc_{\epsilon'/6}\parens*{ \GHD^{C_{\epsilon, \epsilon'}}_{(1/4+\epsilon/2)C_{\epsilon, \epsilon'},(3/4 - \epsilon/2)C_{\epsilon, \epsilon'}} }\ .
  \end{align*}

We conclude by applying known upper bounds for Find the First
Difference~\cite{FeigeRPU1994} (\cref{prop:ftfd}), for
solving many instances of Equality~\cite[Part 6]{FederKNN95}
(\cref{prop:eq}), and Gap Hamming Distance is solved by
exchanging the complete inputs 
which is
essentially optimal~\cite{ChakrabartiR12,Vidick13,Sherstov12}.
\end{proof}



\section{Removing randomness}
\label{app:derandomization}

\subsection{Transcript Distribution Estimation}
\label{app:tde}

In this section we prove \cref{lem:tde-kn97,lem:tde-open}.
\begin{proof}[Proof of \cref{lem:tde-kn97}]
Let $\Pi$ be a communication protocol, and $\gamma = \delta \card{\cT_\pi}^{-1}$.
Given $(x,y)$, the players consider
the protocol tree of $\Pi$ :

\begin{figure}[ht]
  \center
  \includegraphics[page=1]{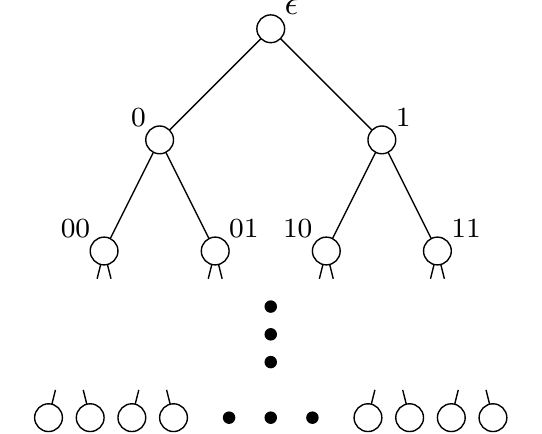}
  \caption{A tree representing the possible executions of the protocol $\Pi$ on a given $(x,y)$.}
  \label{fig:protocol-tree}
\end{figure}

Each node of this tree represents a partial execution of the protocol,
and so we label each node of this tree by the word $w \in \zo^*$ that
is the communication that happened between Alice and Bob to reach this
node. In particular, leaves are labeled by full transcripts, i.e., words
$w \in \cT_\pi$. We will use the notation $w_{<i}$ to refer to the
prefix of $w$ of size $(i-1)$. Each internal node belongs to either Alice or Bob, and that
property determines who must send the next message when at this
specific point of the execution of the protocol. It has $\card{\cT_\pi}$
leaves.

To each internal node $w$, we can assign a probability distribution
$p_w$ that corresponds to which message ($0$ or $1$) is sent
next. This distribution is fully determined by $x$ if the node belongs
to Alice, by $y$ otherwise. Its randomness comes from the private
randomness of the players, and its support is the set of next messages
($0$ or $1$ here). The probability that Alice sends $1$ as her next
message when on node $w$ in the protocol tree is denoted by $p_w(1 \mid x)$.

\begin{figure}[ht]
  \center
  \includegraphics[page=2]{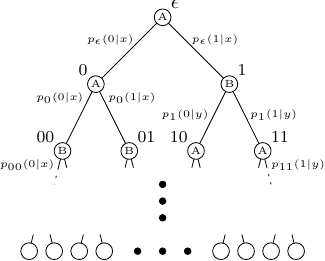}
  \caption{The same tree as in \cref{fig:protocol-tree} with nodes labeled depending on their owners, and the probability distributions. Note that
    $p_w(0 \mid x)+p_w(1 \mid x)=1$.}
\end{figure}

For a leaf of label $w$, on input $(x,y)$, the probability that an execution of the protocol ends up in $w$ is:
\[ p(w \mid x,y) = \underbrace{\parens*{ \prod_{\substack{1 \leq i \leq \card{w}\\w_{< i} \in \textrm{Alice}}} p_{w_{<i}}(w_i \mid x) }}_{\alpha(w \mid x)}
                \quad \times \quad
                   \underbrace{ \parens*{ \prod_{\substack{1 \leq i \leq \card{w}\\w_{< i} \in \textrm{Bob}}} p_{w_{<i}}(w_i \mid y)} }_{\beta(w \mid y)}. \]

For each $w\in \cT_\pi$, Alice has full knowledge of $\alpha(w \mid x)$ and
Bob has full knowledge of $\beta(w \mid y)$. We now describe the actual
protocols in the Bob and in the open model.

\begin{description}
  \item[Step 1.] 
  For each $w\in \cT_\pi$, Alice
sends the smallest non-negative integer $d_w< \ceil*{ \frac 1
\gamma }$ such that:
  \[ \gamma \cdot d_w \leq \alpha(w \mid x) \leq \gamma \cdot (d_w+1).\]

This is done with communication $\card{\cT_\pi} \cdot \ceil*{ \log \frac 1
\gamma }$.

Bob now knows an approximation $\alpha'(w \mid x):=\gamma \cdot d_w$ of
Alice's $\alpha(w \mid x)$ for all $w$ such that:
  \[ \alpha'(w \mid x) \leq \alpha(w \mid x) \leq \alpha'(w \mid x) + \gamma.\]

Since $\beta(w \mid y) \in [0,1]$ for all $w$, $p'(w \mid x,y):=\alpha'(w \mid x)\beta(w \mid y)$
(known to Bob) is such that:
  \[\forall w\in \cT_\pi: p'(w \mid x,y) \leq p(w \mid x,y) \leq p'(w \mid x,y) + \gamma.\]
That is, Bob has an estimation of the true probabilities of $p(. \mid x,y)$
that never overestimates the true value and is pointwise $\gamma$-close to
it.

\item[Step 2.]
Bob cannot simply output
  $p'(. \mid x,y)$ since it might not be a probability
  distribution. However, $p(. \mid x,y)$ is a probability
  distribution, so:
\[1 - \gamma\card{\cT_\pi} \leq \sum_w p'(w \mid x,y) \leq 1.\]
 Let us define $C:=1-\sum_w p'(w \mid x,y)$ and $p''(w \mid x,y) = p'(w \mid x,y) +
 \frac C {\card{\cT_\pi}}$ for all $w$. Since $0\leq C \leq
 \gamma\card{\cT_\pi}$, $p''(. \mid x,y)$ is a distribution which is also a
 point-wise $\gamma$-approximation of $p(. \mid x,y)$. 
 Our choice of $\gamma = \delta \card{\cT_\pi}^{-1}$
 makes $p''(. \mid x,y)$ $\delta$-close to
 $p(. \mid x,y)$ in statistical distance, so Bob can output it. Therefore
  \[ D^\bob(\TDE_{\Pi,\delta}) \leq \card{\cT_\pi} \cdot \ceil*{ \log \frac {\card{\cT_\pi}}{\delta} } \ .\]
 which concludes the proof of \cref{lem:tde-kn97}.
  
\end{description} 

\end{proof}
{
We will show a
  similar statement in the local and open models that we will use in
  proving other derandomization results.

  \begin{lemma}
  \label{lem:tde-open}
  Let $\Pi$ be a private coin communication protocol and $\cT_\pi$ its set of possible transcripts. For any $0<\delta<\frac 1 2$, 
$ D^\loc(\TDE_{\Pi,\delta}) \leq D^\op(\TDE_{\Pi,\delta}) \leq 2 \card{\cT_\pi} \cdot \ceil*{ \log \frac {2 \card{\cT_\pi}}{\delta} }
$.
  \end{lemma}

Note that in the local model, we require that both players output the
same approximation of the distribution on the leaves. 
In the original protocol of \cref{lem:tde-kn97}, it was enough for one player to send estimates and
the other to use its exact values, but here, the second player must also send
back the result.
Hence, the protocol for $\TDE$ has slightly higher communication complexity
in the local model than in the unilateral model.

In the local model, after running the protocol for $\TDE$, both
players have the same estimate for the distribution over the leaves. Each player additionally knows her output distribution on each leaf,
making the majority answer clear for both players. The situation is similar
for an external observer in the open model after openly computing
$\TDE$. Therefore,  following, e.g., the proof of~\cite[Lemma 3.8]{KushilevitzN1997}, \cref{lem:tde-open} 
implies 
\cref{thm:derand-local}.
}

\begin{proof}[Proof of \cref{lem:tde-open}]
Let $\gamma = \frac \delta 2 \card{\cT_\pi}^{-1}$. We proceed as in the proof of \cref{lem:tde-kn97} but we replace the second step by the following one.
\begin{description}
	\item[Step 2'.]
Instead of outputting directly after the first step,
  Bob sends back an approximation of $p'(. \mid x,y)$ to Alice. More precisely,
  for all $w$, he sends $d'_w, 0\leq d'_w < \ceil*{ \frac 1
  \gamma }$ such that:
    \[ \gamma \cdot d'_w \leq p'(w \mid x,y) \leq \gamma \cdot (d'_w+1).\]

  This again takes communication $\card{\cT_\pi} \cdot \ceil*{ \log (\frac 1 \gamma) }$. Hence an external observer knows $p''(w \mid x,y) := \gamma \cdot d'_w $
  for all $w$, which satisfies:
  \[\forall w\in \cT_\pi: p''(w \mid x,y) \leq p(w \mid x,y) \leq p''(w \mid x,y) + 2\cdot \gamma.\]

  Let us define $C:=1-\sum_w p''(w \mid x,y)$ and $p'''(w \mid x,y) = p''(w \mid x,y)
  + \frac C {\card{\cT_\pi}}$ for all $w$. This $p'''(. \mid x,y)$ is a
  distribution, and a $2\cdot \gamma$ point-wise approximation of
  $p(. \mid x,y)$, and can be computed by an external observer. 
  By our choice of $\gamma = \frac \delta 2 \card{\cT_\pi}^{-1}$,
  we get the output we want
  and so 
  \[D^\op(\TDE_{\Pi,\delta}) \leq 2 \card{\cT_\pi} \cdot \ceil*{ \log \frac {2 \card{\cT_\pi}}{\delta} },\]
  which concludes the proof of \cref{lem:tde-open}.
 \end{description}
\end{proof}

\subsection{Proof details for randomness removal (\texorpdfstring{\cref{sec:derandomization}}{Section~\ref{sec:derandomization}})}
\label{app:proofs-derandomization}


The following lemma will be useful in proving our results:

\begin{lemma}
\label{lem:leafs-and-outputs}
Let $U$ and $V$ be random variables over their respective domain $\cU$
and $\cV$. For all $u\in \cU$, le us consider $V_{U=u}$ the random
variable $V$ conditioned on the event $[U=u]$. Assume there exists
two constants $\delta_U$ and $\delta_V$ and two random variables $U'$
and $V'$ over the same domains as $U$ and $V$ such that:
\[\Delta(U,U')\leq \delta_U \qquad \forall u\in \cU: d_\infty(V_{U=u},V'_{U'=u}) \leq \delta_V.\]
Then:
\[d_\infty(V,V')\leq \delta_U + \delta_V.\]
\end{lemma}

\begin{proof}[Proof of \cref{lem:leafs-and-outputs}]
Let us show that $\forall v \in \cV$, $\abs*{\Pr[V=v] -
\Pr[V'=v] } \leq \delta_U + \delta_V$. Fix an arbitrary
$v\in \cV$, then the probabilities $\Pr[V=v]$ and $\Pr[V'=v]$ can be
written as:
\begin{itemize}
\item $\Pr[V=v]  = \sum_{u \in \cU} \Pr[U=u]\cdot \Pr[V=v \mid U=u]$,
\item $\Pr[V'=v] = \sum_{u \in \cU} \Pr[U'=u]\cdot \Pr[V'=v \mid U'=u]$.
\end{itemize}

Hence using our two hypotheses above we get:
\begin{eqnarray*}
 \lefteqn{\Pr[V=v] - \Pr[V'=v]} \\
 &= & \sum_{u \in \cU} \parens*{
                \Pr[U=u]\cdot \Pr[V=v \mid U=u] - \Pr[U'=u]\cdot \Pr[V'=v \mid U'=u]
                } \\
 &\leq& \sum_{u \in \cU} \parens*{\parens*{\Pr[U=u] - \Pr[U'=u] }
                                \Pr[V=v \mid U=u] + \delta_V \Pr[U'=u] }\\
 &\leq &\sum_{u \in \cU:\Pr[U=u] > \Pr[U'=u]}
                      \parens*{\Pr[U=u] - \Pr[U'=u] } 
                      + \delta_V \\
 &\leq & \delta_U + \delta_V.            
\end{eqnarray*}

We can prove $\Pr[V=v] - \Pr[V'=v] \geq -(\delta_U + \delta_V)$
following the same proof method, and combining the two we get the
desired result:
\[\forall v\in\cV: \abs*{\Pr[V=v] - \Pr[V'=v]} \leq \delta_U + \delta_V.\]
\end{proof}

\subsection{Derandomization up to the XOR model}

The open and local models are straightforward adaptations of \cref{thm:derand-kn97}.

\begin{theorem}
For any function $f$, error $\epsilon < \frac 1 2$ and model $\mdl \in \set{\op,\loc}$, with $R^\mdl = R^{\mdl,\priv}_\epsilon(f)$:
  \label{thm:derand-local}
\[
D^\mdl(f) \leq 2\cdot 2^{R^\mdl}
	\parens*{ 
	R^\mdl +
	\log \parens*{\tfrac {1}{\frac 1 2 - \epsilon}} + 2 
	}
\]
\end{theorem}

\subsubsection{Derandomization in the one-out-of-two model}

Interestingly, in the one-out-of-two model, there is an error threshold for
derandomization at $\epsilon = \frac 1 3$. 
If the error is below this threshold, solving the appropriate instance
of $\TDE$ suffices, after which one of the players knows the majority
outcome.  When the error is close to $1/2$, there can be several
candidates for the majority outcome, which would cost  an additional $O(k)$
to communicate.
We reduce this term to $O(\log(k))$ in this case by using a variant of the
NBA problem. 

\begin{theorem}
  \label{thm:derand-oot}
For any function $f$ and error $\epsilon < \frac 1 2$, with $R = R^{\oot,\priv}_\epsilon(f)$:
\[  D^\oot(f) \leq 
\begin{cases}
        2^{R+1} \parens*{R + \log \parens*{\frac 4 {\frac 1 3 - \epsilon }} + 1},  
		& \text{if }  \epsilon < \frac 1 3, \\
\parens*{  2^{R+1} + 2 }
    \cdot \parens*{R + \log \parens*{\frac 8 {\frac 1 2 - \epsilon } }
    + 1 } + \log(k)    + 4,
		& \text{for any }   \epsilon < \frac 1 2.
\end{cases}
\]

\end{theorem}

\begin{proof}[Proof of \cref{thm:derand-oot}]
    Take $\Pi$ to be an optimal private coin one-out-of-two protocol for $f$ with
    error $\epsilon$. Let $\sigma$ be a precision parameter which we
    will set later.
    
    When $\epsilon < \frac{1}{3}$, notice that one of the players has
    to output the correct result with probability greater than $\frac
    1 3$, while all incorrect ones are output with probability less
    than $\frac 1 3$ (with an additional small bias). So it suffices
    for the players to run the local protocol of
    \cref{lem:tde-open} for $\TDE_{\pi,\sigma}$ where $\sigma
    < \frac 1 3 - \epsilon$ in this case, and let the player who
    outputs some result with probability greater than $\frac 1 3$
    output it.

  We now turn to the more interesting case where $1/3\leq \epsilon <
  \frac{1}{2}$. Let $\delta = \frac 1 2 -\epsilon$ and $\sigma < \frac \delta 3$. The players first
  run the local protocol for $\TDE_{\pi,\sigma}$, thus learning a
  $\sigma$ approximation of the probability of each transcript of the
  protocol.  By \cref{lem:leafs-and-outputs},
    since each player exactly knows her outputting distribution in
    each leaf, for all $z$, each player knows up to precision $\sigma$
    her probability of outputting $z$ in the original protocol.

  Let us call $p_\ali^z$ the probability that Alice outputs $z$, and
  $\widetilde p_\ali^z$ the approximation she has of it. For $z =
  f(x,y)$, we have $p_\ali^z+p_\bob^z \geq \frac 1 2 + \delta$ and so
  $\widetilde p_\ali^z + \widetilde p_\bob^z \geq \frac 1 2 + \delta - \sigma$.

  Using this, the players consider some $z$ as \emph{candidates} for
  $f(x,y)$. Alice considers $(z^\ali_i)_{i \in [n_\ali]}$ the $n_\ali$ answers
  $z$ such that $\widetilde p_\ali^z \geq \frac 1 4 + \frac {\delta -
    \sigma} 2$.  Similarly, Bob considers $(z^\bob_j)_{j \in [n_\bob]}$ the $n_\bob$ answers
  $z$ such that $\widetilde p_\bob^z \geq \frac 1 4 + \frac {\delta -
    \sigma} 2$.

  Since $\sum_z \widetilde p_\ali^z + \widetilde p_\bob^z = 1$ (where the sum is over all $z \in \cZ$), we have that:
  $n_\ali + n_\bob \leq 3$. Since the majority output represents strictly more than half 
  of all ouputs we have $\max(n_\ali,n_\bob) \leq 2$.

  The players use $4$ bits to send the values $n_\ali,n_\bob$ to each
  other. Without loss of generality, assume $n_\ali \geq n_\bob$. Then four
  cases are possible:
  \begin{enumerate}
  \item $(n_\ali,n_\bob) = (1,0)$ \label[case]{case-1-cand}
  \item $(n_\ali,n_\bob) = (2,1)$ \label[case]{case-3-cands}
  \item $(n_\ali,n_\bob) = (2,0)$
  \item $(n_\ali,n_\bob) = (1,1)$.
  \end{enumerate}

  The first two cases are simple: if there is only one candidate
  (\cref{case-1-cand}), the player who owns it outputs it. If
  there are three candidates (\cref{case-3-cands}), the player
  with a single candidate outputs it knowing that it has to match one
  of the candidates on the other side and be the majority output.

  For the remaining two cases, we will use a variant of the protocol
  for the NBA problem.
   For the case $(n_\ali,n_\bob) = (2,0)$, 
  Alice (who has two candidates)
  sends to Bob the index of a bit where the two
  candidates differ, say $i \in [\ceil{\log(\cZ)}]$. Bob replies with
  $\sum_{z : z_i = 0} \widetilde p_\bob^z$. 
  Alice can thus compute
  $\sum_{z : z_i = 0} \widetilde p_\bob^z + \widetilde p_\ali^z$. If that
  quantity is greater than $\frac 1 2$, the correct candidate is the
  one whose $i$-th bit is $0$; otherwise, it is the other candidate.
  
  Finally,   
  let us consider the case $(n_\ali,n_\bob) = (1,1)$. Without loss of generality, assume Alice's candidate, $z^\ali_1$, is not correct, that is, $z^\ali_1 \neq f(x,y) = z^\bob_1$. 
  Then, we notice that the
  probability Alice outputting $z^\ali_1$ and the probability of Bob outputting something different from $z^\bob_1$ are less than $\epsilon = \frac
  1 2 - \delta$. 
%
%
  To conclude the protocol, the players exchange $\widetilde p_\ali^{z^\ali_1}$ and
  $\widetilde p_\bob^{z^\bob_1}$ up to $\sigma$ precision. 
  Then:
  \begin{itemize}
  \item $p_\bob^{z^\bob_1} + p_\bob^{\top} - p_\ali^{z^\ali_1} = p_\bob^{z^\bob_1} +
    \sum_{z \neq z^\ali_1} p_\ali^z \geq p_\bob^{z^\bob_1} + p_\ali^{z^\bob_1} \geq \frac
    1 2 + \delta$,
  \item $p_\ali^{z^\ali_1} + p_\ali^{\top} - p_\bob^{z^\bob_1} = p_\ali^{z^\ali_1} +
    \sum_{z \neq z^\bob_1} p_\bob^z \leq 1 - p_\ali^{z^\bob_1} + p_\bob^{z^\bob_1}\leq
    \frac 1 2 - \delta$.
  \end{itemize}

  Each player has a $\sigma$ approximation of the sum of probabilities
  of outputs on her side, and a $2\sigma$ approximation of the
  probability of the candidate output on the other player's side, so
  they have a $3\sigma$ approximations of the above sums. Since
  $\sigma < \frac \delta 3$, the players know with
  certainty if they have the correct output or not. If they do not have the correct output, they
  let the other player output.
\end{proof}

\subsubsection{Derandomization in the split model}

Derandomization in the split model can be achieved 
similarly to derandomization in the previously studied models.

\begin{theorem}
        \label{thm:derand-split}
   Let $0<\epsilon <1/2$ and $f : \cX \times \cY \rightarrow \cZ = \zo^k$. Let $R=R^{\spt,\priv}_\epsilon(f) $, $M=16 \cdot \parens*{\frac 1 2 - \epsilon }^{-2} \cdot 2^{R}$. Then:
\[
D^\spt(f) \leq
\begin{cases}
        2^{R+1} \cdot \parens*{R + \log \parens*{\frac 4 {\frac 1 3 - \epsilon }} + 1} + k,  
		& \text{if }  \epsilon < \frac 1 3, \\
   2^{R+1}  \cdot \parens*{R + \log \parens*{\frac 8 {\frac 1 2 - \epsilon }} + 1} 
+ k \cdot \parens*{\frac {5-2\epsilon} 4 M+1},
		& \forall   \epsilon < \frac 1 2.
\end{cases}
\]
\end{theorem}

\begin{proof}[Proof of \cref{thm:derand-split}]
  \begin{description}
  \item[Case when $\epsilon < \frac 1 3$ ] As in the proof of
    \cref{thm:derand-oot}, for each string $z\in\set{0,1,*}^k$
    the players estimate their probability of outputting $z$ by
    solving a $\TDE$ instance. For each index of the output, in the randomized protocol, one of the players has to output the correct bit with probability at least $\frac {1-\epsilon} 2 > \frac 1 3$. Alice sends $k$ bits to Bob to indicate for which bits she outputs the same non-$*$ symbol with probability more than $\frac 1 3$ in the original protocol. She outputs those bits in the derandomized protocol, while Bob is in charge of outputting the other bits. For each bit they output, they output the value which was most frequent in the original randomized protocol.
  \item[Case when $\epsilon < \frac 1 2$] This case is similar to what we
    saw in the proof of \cref{thm:derand-xor}: we create a
    $\GapMAJ$ composed with the weave ($\weave$) gadget, in the same way that we created a $\GapMajX$
    instance to derandomize a protocol in the XOR model. 
  \end{description}
\end{proof}

As in the XOR model, these bounds can be improved by improving the deterministic complexity of the right gadgetized version of $\GapMAJ$ in the split model.


\section{Proofs of the bounds on \texorpdfstring{$\GapMajX$}{GapMAJ XOR}}
\label{app:cc-hxor}

In this section we prove the statements of \cref{tab:cc-hxor}
about the communication complexity of $\GapMajX$.

{Recall that the
$\GapMajX$ problem is parameterized by four parameters: $N$ the number of
rows, $k$ the length of Alice's and Bob's rows, $\epsilon$ the
fraction of rows that do not XOR to the hidden $k$-bit string $z$, and
$\mu$ a distribution over the rows. A $\GapMajX$ instance can be pictured
as Alice and Bob each having a $N \times k$ boolean matrix such that
the set of rows containing $z$ in the bitwise XOR of the two matrices
has a weight higher than $1-\epsilon$ relative to $\mu$. In what
follows, $\epsilon'$ is the target probability, i.e., the maximal
error rate we tolerate when solving our $\GapMajX$ instances.}

When $\epsilon \leq \epsilon'$, the trivial protocol in the XOR model
which consists in choosing a row with public coins and outputting that row
suffices. We thus have the following result.
\begin{proposition}
\label{prop:bdd-xor-pub-epsprime-HXOR}
For all $N,k,\epsilon,\epsilon',\mu$, with $0\leq \epsilon\leq \epsilon'$, 
\[R_{\epsilon'}^{\xor,\pub} (\GapMajX[N,k,\epsilon,\mu]) = 0. \]
\end{proposition}
From this proposition we derive the following upper bounds.
\begin{corollary}\label{cor:simple-hxor-bounds}
For all $N,k,\epsilon,\epsilon',\mu$, with $0\leq \epsilon\leq \epsilon'$ 
\begin{itemize}
\item $R_{\epsilon'}^{\xor,\priv} (\GapMajX[N,k,\epsilon,\mu]) \leq \log (N) $,
\item $R_{\epsilon'}^{\op,\pub} (\GapMajX[N,k,\epsilon,\mu]) \leq 2 k $,
\item $R_{\epsilon'}^{\op,\priv} (\GapMajX[N,k,\epsilon,\mu]) \leq 2 k + \log (N) $,
\item $D^{\uni}(\GapMajX[N,k,\epsilon,\mu]) \leq (2\epsilon N+1) k$.
\end{itemize}
\end{corollary}
\begin{proof}[Proof of \cref{cor:simple-hxor-bounds}]
The key is to notice  that  $\GapMajX$ is trivial when $\epsilon\leq \epsilon'$ (\cref{prop:bdd-xor-pub-epsprime-HXOR}). 
For private coins, notice that the players only use $\log(N)$ coins,
so it is enough for Alice (w.l.o.g.) to send her coins to Bob.
For an open protocol, the players can exchange their rows.  For the
deterministic protocol, Alice (w.l.o.g.) sends her $(2\epsilon N+1)$
heaviest rows ($\mu$-wise) to Bob, who can then compute
the most frequently occurring $z$ to which his rows and Alice's rows XOR. 
\end{proof}

 When $\epsilon > \epsilon'$, we refer to our earlier result from
\cref{sec:error-reduction} (\cref{thm:hxor-no-k}).



\section{Other separating problems}
\label{app:separating-problems}


In this appendix, we give the definitions of the separating problems in
\cref{fig:hierarchy,tab:all-problems} and prove that their complexity depends
on the output model. All of them are variations of common problems
with an additional constraint over the inputs, namely that at most $t$
bits of each player's $n$-bit input are ones. Let us denote by 
$B_2(n,t) = \set{x \in \zo^n : \sum_i x_i \leq t}$ the
Hamming ball of radius $t$ in $\zo^n$ centered at $0^n$,
 by $H(x)=-x\log(x)-(1-x)\log(1-x)$ the entropy of
a Bernouilli random variable of expected value $x$, and recall the
following bound on its size, which we denote by $V_2(n,t)=\card{
B_2(n,t) }$:
\begin{lemma}[Chapter 10, Corollary 9 in~\cite{MWS83}]
  \label{lem:hamming-ball}
  Let $0 < t < n/2$. Then:
  \[ \frac 1 {\sqrt{8 t (1- t/n)}} 2^{n\cdot H(t/n)} \leq V_2(n,t) \leq 2^{n\cdot H(t/n)}. \]
  
\end{lemma}
  
In what follows, we will consider $t \in o(n)$, and only use that in
this regime:
\[\log \parens{V_2(n,t)} \in \Omega(t\cdot \log(n)).\]

\subsection{\texorpdfstring{$t$-Intersection}{t-Intersection}}~
\label{sec:t-INT}

Since Disjointness is a Boolean problem, it cannot separate our models
of communication. It is not the case, however, of its large-output
variant Intersection, where Alice and Bob must compute the actual
intersection of their sets.

We recall the definitions of the problems $t-\DISJ_n$ and $t-\INT_n$,
what is known about their complexities, and show that $t-\INT_n$
separates the local model from the open model.

\begin{definition}[$t$-Disjointness problem]
  $t-\DISJ_n : B_2(n,t) \times B_2(n,t) \rightarrow \zo$ is defined as:
\[ t-\DISJ_n(X,Y) =  {\indic}_{X \cap Y = \emptyset}. \]

\end{definition}

We now define a natural variation of this problem, with large output.

\begin{definition}[$t$-Intersection problem]
\label{def:t-inter}
  $t-\INT_n : B_2(n,t) \times B_2(n,t) \rightarrow B_2(n,t)$ is defined as:
  \[  t-\INT_n(X,Y)=X \cap Y. \]

\end{definition}

Since the output of $t-\DISJ_n$ is boolean, its various communication
complexities are essentially the same up to one bit so we do not need
to specify the communication model in the following statement:

\begin{theorem}
  $ R_\epsilon(t-\DISJ_n) = \Theta(t).$
\end{theorem}

The $\Omega(t)$ lower bound comes directly from the $\Omega(n)$ lower
bound for $\DISJ_n$
of~\cite{KalyanasundaramS92,Razborov92,MR2059642-BarYossefJKS-2004},
while the $O(k)$ upper bound was proven in~\cite{HastadW2007}.

\begin{theorem}
  $ R^\loc_\epsilon(t-\INT_n) = \Theta(t)$, and  $R^\op_\epsilon(t-\INT_n) = \Theta(t \cdot \log(n)).$
\end{theorem}

The $O(t)$ upper bound for this problem was proved
in~\cite{BrodyCKWY14} and the $\Omega(t \cdot \log(n))$ lower bound in the
open model simply comes from the size of the output (\cref{app:wprt} and \cref{lem:hamming-ball}) since $\card{B_2(n,t)} = V_2(n,t) \in \Omega(t \cdot \log(n))$ (for $t \in o(n)$).

\subsection{\texorpdfstring{$t$-Find the First Difference}{t-Find the First Difference}}~
\label{sec:FTFD}

Just as Intersection can be seen as a large-output variant of the
Disjointness problem, Find the First Difference can be thought of as
the large-output variant of the Greater Than problem.

We now define the problems $t-\GT_n$ and $t-\INT_n$, what is known
about their complexities, and show that $t-\FtFD_n$ separates the
one-out-of-two model from the unilateral model.

\begin{definition}[$t$-Greater Than problem]
  $t-\GT_n : B_2(n,t) \times B_2(n,t) \rightarrow \zo$ is defined as:
\[ t-\GT_n(x,y) = {\indic}_{x > y}.  \]
\end{definition}

\begin{definition}[$t$-Find the First Difference problem]
  $t{-}\FtFD_n : B_2(n,t) \times B_2(n,t) \rightarrow \set{0,\ldots,n}$ is
  defined as:
\[ t-\FtFD_n(x,y) = \min \parens*{ \set{i : x_i \neq y_i} \cup \set{n} }.  \]
\end{definition}

\begin{theorem} \label{thm:t-ftfd-gap}
  \begin{align*}
      R^\oot_\epsilon(t-\FtFD_n) &\in O\parens*{\log (t) + \log (\log (n)) + \log \parens*{\frac 1 \epsilon} }, \\
 R^\uni_\epsilon(t-\FtFD_n) &\in \Omega(\log (n)).
  \end{align*}
\end{theorem}

\begin{proof}[Proof of \cref{thm:t-ftfd-gap}]~

  \begin{description}
  \item[Upper bound on $R^\oot_\epsilon(t-\FtFD_n)$.]
  As an intuition, let us first give a protocol in the case $t=1$.

  In this case, Alice and Bob $n$-bit strings $x$ and $y$ only contain
  a single $1$ each. So consider $i^\ali, i^\bob$ such that $x_{i^\ali} = 1$
  and $y_{i^\bob} = 1$. $i^\ali,i^\bob \in [n]$, therefore they can be written
  as two $\ceil{ \log n }$-bit strings.

  The players then run the protocol of Feige et al~\cite{FeigeRPU1994}
  to find the first difference between $i^\ali$ and $i^\bob$. Doing so, they
  learn the smallest $t$ such that $(i^\ali)_k \neq (i^\bob)_k$ (or
  $\ceil{\log(n)}+1$ if it does not exist), and so whether $i^\ali <
  i^\bob$, $i^\ali > i^\bob$ or $i^\ali = i^\bob$. The player that has the lowest
  number thus knows the index of the first difference between $x$ and
  $y$, as it is $\min(i^\ali, i^\bob)$.

  Now consider $t$ unconstrained. To find the first difference between
  their two $n$-bit strings of weight $\leq t$, the two players simply
  construct a $\Omega(t \cdot \log(n))$-bit string made of the indices of
  their $1$ bits (with adequate padding) and use the protocol of Feige
  et al~\cite{FeigeRPU1994} as in the $t=1$ case. More precisely:

  \begin{itemize}
    \item Let $w_x = \card{x} \leq t$ (resp.\ $w_y = \card{y} \leq t$) be the
      weight of $x$ (resp.\ $y$). Now, consider indices $i^\ali_1, \ldots,
      i^\ali_{t}$ and $i^\bob_1, \ldots, i^\bob_{t}$, in $\set{0,\ldots,n-1} \cup
      \set{ 2^{\ceil{ \log (n+1) }} - 1 }$ such that:

      \begin{itemize}
      \item $i^\ali_j = 2^{\ceil{\log (n+1)}} - 1$ (an all-$1$
        string) iff $j > w_x$
      \item $x_{i^\ali_j} = 1, \forall j <= t$
      \item $i^\ali_j < i^\ali_{j+1}, \forall j < t$

        (and similarly for the $i^\bob_j$'s)
      \end{itemize}
      
      Each $i^\ali_j$ can be written on $\ceil{ \log (n+1) }$ bits,
      so Alice computes a $t \ceil{ \log n }$-bit string $s_x$
      made of the concatenation of all the $i^\ali_j$'s, in order. Bob
      computes $s_y$ similarly.

      Then the two players use the protocol of Feige et al to obtain
      the first difference between $s_x$ and $s_y$. Let us note
      $i_{\diff}$ the index of this difference.

      Then Alice knows the index of the first difference if
      $(s_x)_{i_\diff} = 0$, and otherwise Bob does. Indeed, let us
      consider the first case:

      \begin{itemize}
      \item The fact that there is a $0$ on this index for Alice means
        that this part of $s_x$ corresponds to the position of a $1$
        in the original $n$-bit string $x$, since we pad with $1$'s at
        the end.
      \item This position is the index of the leftmost $1$ that Alice
        has but Bob does not have. Indeed, all positions before the
        one $i_\diff$ belongs to are shared between Alice and Bob. So
        if Bob also had a $1$ in the position in which $i_\diff$
        appears, then the fact that Alice and Bob find a difference in
        $i_\diff$ means that Bob also has a $1$ in a smaller
        position, which contradicts the fact that the first difference
        between $s_x$ and $s_y$ was such that Alice has a $0$ at that
        place.
      \end{itemize}

      Using Feige et al's protocol on a $O(t \cdot \log (n))$-bit string  costs $O\parens*{\log \parens*{ \frac {t \cdot \log (n)} {\epsilon}}}$, hence
      the advertised upper bound.
  \end{itemize}

  \item[Lower bound on $R^\uni_\epsilon(t-\FtFD_n)$.]
  Let Alice be the outputting player (w.l.o.g.), and consider inputs
  where she always receives the all-$0$ $n$-bit string and Bob
  receives a random $n$-bit string with a single $1$. Solving Find the
  First Difference on such instances would allow Bob to send an
  information of size $\log(n)$ bits to Alice with
  $R^\ali_\epsilon(\FtFD_n)$ communication and high probability, hence
  the $\Omega(\log(n))$ lower bound.
  \end{description}
\end{proof}

Note that our one-out-of-two derandomization theorem
(\cref{thm:derand-oot}) shows that our upper bound is tight for
private coin communication complexity, but it may still be that there
is a more efficient public coin protocol in the one-out-of-two or the
XOR model. We now show that Viola's $\Omega(\log(n))$ public coin
randomized lower bound~\cite{Viola2015} for $\GT_n$ implies that this
protocol is also tight when given access to public coins.

\begin{theorem}\label{thm:t-gt-lb}
  \[R_\epsilon(t-\GT_n) \in \Omega(\log(t) + \log\log(n))\]
  and as a corollary, $R_\epsilon^\xor(t-\FtFD_n) \in \Omega(\log(t) + \log\log(n))$.
\end{theorem}

\begin{proof}[Proof of \cref{thm:t-gt-lb}]
  We prove the dependencies in $\log(t)$ and in $\log \log(n)$ independently.

\begin{description}
      \item[$R_\epsilon(t-\GT_n) \in \Omega(\log(t))$.] We remark that
  $\GT_t$ reduces to $t-\GT_n$ in the same way that $\DISJ_t$ reduced
  to $t-\DISJ_n$ in the previous section, so applying Viola's lower
  bound~\cite{Viola2015} yields:
  \[R_\epsilon(t-\GT_n) \geq R_\epsilon(\GT_t)\in \Omega(\log(t))\]
  
  \item[$R_\epsilon(t-\GT_n) \in \Omega(\log \log (n))$.] We remark
  that $1-\GT_n$ reduces to $\GT_{\log(n)}$ since a way to compare two
  numbers with a single bit set to one in their binary representation
  is to compare the indices of the position of their single
  one. Hence, applying Viola's lower bound~\cite{Viola2015} again:
  \[R_\epsilon(t-\GT_n) \geq R_\epsilon(1-\GT_n) \geq R_\epsilon(\GT_{\log(n)})\in \Omega(\log\log(n))\]
\end{description}
\end{proof}

\subsection{The \texorpdfstring{$\MAX$}{MAX} problem}
\label{sec:max}

\begin{definition}[Maximum problem]
\label{def:max}
    $\MAX_n : \zo^n \times \zo^n \rightarrow \zo^n$ is
defined as \[\MAX_n (x,y) =
\begin{cases}
x, & \text{if $x\geq y$,}\\
y, & \text{otherwise.} 
\end{cases}\]
\end{definition}

For this problem, we have: 

\begin{theorem}\label{thm:max-gap}
\[  R^\oot_\epsilon(\MAX_n) \in O(\log n), \qquad
    R^\uni_\epsilon(\MAX_n) \in \Omega(n). \]

The gap is the same (asymptotically, up to multiplicative and additive constants) when only allowing private coins.
\end{theorem}

\begin{proof}[Proof of \cref{thm:max-gap}]~

  \begin{description}
  \item[$R^\oot$ upper bound.]
  The players compute whether $x \leq y$ or not with high probability
  using $O(\log n)$ communication, then if $x \leq y$ Alice outputs~$x$,
  otherwise Bob outputs~$y$.

  \item[$R^\uni$ lower bound.] it suffices to show that $R^\ali$ is
  large, as symmetry will imply that therefore $R^\bob$ is large as
  well.

  The proof is quite simple: consider the  $2^n$ input pairs
  $\set{(0,y) : y \in [0,2^n-1]}$. For those inputs, the
  $\MAX_n$ problem is just a problem of one-way communication: it is
  clear that he must send $\Omega(n)$ bits for Alice to correctly
  guess his $y$ with probability $\geq 1-\epsilon$.
\end{description}
\end{proof}


\end{document}